%% file: QCmax.tex
\documentclass[a4paper]{article}%
\usepackage[utf8]{inputenc}

\usepackage[sort&compress,numbers]{natbib}
\usepackage{csquotes}
\usepackage[main=UKenglish, ngerman]{babel}
\usepackage{fullpage, parskip}
\IfFileExists{Hauke.tex}{\input{Hauke.tex}}{} %

\usepackage{amsmath, amssymb, amsthm}
\usepackage{thm-restate}

\usepackage[appendix=inline,bibliography=common]{apxproof}

\usepackage{hyperref}
\hypersetup{colorlinks=true,linkcolor=red,citecolor=blue}
\newtheoremrep{lemma}[lem]{Lemma}
\newtheoremrep{theorem}[lem]{Theorem}
\newtheoremrep{proposition}[lem]{Proposition}
\newtheoremrep{corollary}[lem]{Corollary}
\newtheoremrep{definition}[lem]{Definition}

\usepackage{bm}

\usepackage{cryptocode}
\createprocedureblock{pcb}{center,boxed}{}{}{linenumbering}

\usepackage{tikz}
\usetikzlibrary{decorations.pathreplacing,calligraphy, calc}
\usepackage{gnuplot-lua-tikz}

\usepackage{booktabs}
\usepackage{array}
\newcolumntype{A}{ >{$} r <{$} @{} >{${}} l <{$} } %
\newcolumntype{x}[1]{>{\centering\arraybackslash\hspace{0pt}}p{#1}}

\makeatletter
\newcommand{\oplabel}[1]{\refstepcounter{equation}(\theequation\ltx@label{#1})}
\makeatother

\usepackage{todonotes}
\presetkeys{todonotes}{inline}{}

\newcounter{dummypcb} %

\usepackage{float}
\newfloat{printalgo}{htb!}{Algorithm}
\floatname{printalgo}{Algorithm}

\makeatletter
\let\savedpcb\pcb
\RenewDocumentCommand{\pcb}{omm}{
    \IfNoValueTF{#1}
        {\savedpcb{#2}{\pcb@addlabel{#2}#3}}
        {\savedpcb[#1]{#2}{\pcb@addlabel{#2}#3}}
}
\newcommand{\pcb@addlabel}[1]{
    \def\thedummypcb{\unexpanded{\unexpanded{#1}}}
    \refstepcounter{dummypcb}
}
\AtBeginDocument{\let\@pc@original@label\label}
\makeatother

\newcommand*{\doi}[1]{\href{https://doi.org/\detokenize{#1}}{doi: \detokenize{#1}}}

\newcommand{\UB}{\operatorname{UB}}
\newcommand{\for}{\operatorname{ for }}
\newcommand{\free}{\operatorname{free}}
\newcommand{\OPT}{\operatorname{OPT}}
\newcommand{\supp}{\operatorname{supp}}

\title{New Support Size Bounds for Integer Programming,\\ Applied to Makespan Minimization on Uniformly Related Machines}

\date\relax
\author{
  Sebastian Berndt\thanks{University of L{\"u}beck, Institute for Theoretical Computer Science, L{\"u}beck, Germany. \texttt{s.berndt@uni-luebeck.de}} \and 
  Hauke Brinkop\thanks{Kiel University, Kiel, Germany. \texttt{hab@informatik.uni-kiel.de}. Partially supported by the DFG Project \href{https://gepris.dfg.de/gepris/projekt/453769249}{\enquote{\foreignlanguage{ngerman}{Fein-granulare Komplexit{\"a}t und Algorithmen f{\"u}r Scheduling und Packungen}}}, \texttt{JA 612 /25-1} } \and 
  Klaus Jansen\thanks{Kiel University, Kiel, Germany. \texttt{kj@informatik.uni-kiel.de}. Partially supported by the DFG Project \href{https://gepris.dfg.de/gepris/projekt/453769249}{\enquote{\foreignlanguage{ngerman}{Fein-granulare Komplexit{\"a}t und Algorithmen f{\"u}r Scheduling und Packungen}}}, \texttt{JA 612 /25-1}} \and
  Matthias Mnich\thanks{Hamburg University of Technology, Institute for Algorithms and Complexity, Hamburg, Germany. \texttt{matthias.mnich@tuhh.de}} \and 
  Tobias Stamm\thanks{Hamburg University of Technology, Institute for Algorithms and Complexity, Hamburg, Germany. \texttt{tobias.stamm@tuhh.de}}
}

\let\eps=\varepsilon

\begin{document}

\maketitle

\begin{abstract}
  Mixed-integer linear programming (MILP) is at the core of many advanced algorithms for solving fundamental problems in combinatorial optimization.
  The complexity of solving MILPs directly correlates with their support size, which is the minimum number of non-zero integer variables in an optimal solution.
  A hallmark result by Eisenbrand and Shmonin (\emph{Oper.\@ Res.\@ Lett.\@}, 2006) shows that any feasible integer linear program (ILP) has a solution with support size $s\leq 2m\cdot\log(4m\Delta)$, where~$m$ is the number of constraints, and~$\Delta$ is the largest coefficient in any constraint.
  
  Our main combinatorial result are improved support size bounds for ILPs.

  To improve granularity, we analyze for the largest $1$-norm $A_{\max}$ of any column of the constraint matrix, instead of $\Delta$.
  We show a support size upper bound of $s\leq m\cdot(\log(3A_{\max})+\sqrt{\log(A_{\max})})$, by deriving a new bound on the -1 branch of the Lambert $\mathcal{W}$ function.
  Additionally, we provide a lower bound of $m\log(A_{\max})$, proving our result asymptotically optimal.
  Furthermore, we give support bounds of the form $s\leq 2m\cdot\log(1.46A_{\max})$.
  These improve upon the previously best constants by Aliev. et. al. (\emph{SIAM J. Optim.} 2018), because all our upper bounds hold equally with $A_{\max}$ replaced by $\sqrt{m}\Delta$.

  Our main algorithmic result are the fastest known approximation schemes for fundamental scheduling problems, which use the improved support bounds as one ingredient.
  
  First, we design an efficient approximation scheme (EPTAS) for makespan minimization on uniformly related machines ($Q||C_{\max}$).
  Our EPTAS yields a $(1+\eps)$-approximation for $Q||C_{\max}$ on~$N$ jobs in time $2^{\mathcal{O}(1/\eps\log^3(1/\eps)\log(\log(1/\eps)))}+\mathcal{O}(N)$, which improves over the previously fastest algorithm by Jansen, Klein and Verschae (\emph{Math. Oper. Res.}, 2020) with run time $2^{\mathcal{O}(1/\eps\log^4(1/\eps))}+N^{\mathcal{O}(1)}$.
  Arguably, our approximation scheme is also simpler than all previous EPTASes for $Q||C_{\max}$, as we reduce the problem to a novel MILP formulation which greatly benefits from the small support.

  Second, we consider $Q|\mathit{HM}|C_{\max}$, which is the problem $Q||C_{\max}$ with high-multiplicity encoding where \emph{both} jobs and machines are given as few distinct types according to their processing times and speeds.
  Until now, not even a constant-factor approximation was known for this problem.
  We give the first EPTAS for $Q|\mathit{HM}|C_{\max}$,  with run time ${2^{\mathcal{O}(1/\eps\log^3(1/\eps)\log(\log(1/\eps)))}+\langle\mathcal{I}\rangle^{\mathcal{O}(1)}}$. %
  
  Third, we consider the problem $R_KQ||C_{\max}$, where each job can have up to $K$ distinct processing times, depending on the type of machine they get processed on, and machine speeds are uniformly related.
  We give an EPTAS for $R_KQ||C_{\max}$ with run time ${2^{\mathcal{O}(K\log(K)1/\eps\log^3(1/\eps)\log(\log(1/\eps)))}+\mathcal{O}(K N)}$, thereby reducing the previously best exponent $\tilde{\mathcal{O}}(K\cdot 1/\eps^3)$ by Jansen and Maack (\emph{Algorithmica}, 2019) down to $\tilde{\mathcal{O}}(K\cdot 1/\eps)$.
\end{abstract}

\section{Introduction}
\label{sec:introduction}
The {\sc Integer Linear Programming} (ILP) problem is to find an optimal integral solution $\bm{x}^\star\in\mathbb Z^n_{\geq0}$, which minimizes a linear objective $\bm{c}^\intercal \bm{x}$ subject to a system $A\bm{x}=\bm{b}$ of $m$ linear constraints.
In the {\sc Mixed-Integer Linear Programming} (MILP) problem, the solution sought is a pair $(\bm{x}^\star,\bm{y}^\star)\in \mathbb Z^n_{\geq0}\times\mathbb Q^r_{\geq0}$, which minimizes $\bm{c}^\intercal(\bm{x}~\bm{y})$ subject to $A\bm{x} + B\bm{y}=\bm{b}$.
Solving (M)ILPs is at the core of many advanced algorithms for fundamental combinatorial optimization problems.
Very often, the run time of these algorithms scales with the \emph{support size} of the (M)ILPs, which is the smallest number of non-zero entries in an optimal solution~$\bm{x}^\star$ resp.\@ $(\bm{x}^\star,\bm{y}^\star)$.
Thus, support size bounds have found applications all over computer science, for example in scheduling~\cite{BansalOVZ2016,JansenKV2016}, logic~\cite{KuncakR2007,Pratt-Hartmann2008}, and even complexity theory~\cite{HaaseZ2019}.
Therefore, the smaller the support size, the better these results become.
Hence, an important research direction is to prove strong upper bounds on the support size of (M)ILPs.
The original result on support size bounds, which is already finding its way into the standard curriculum of integer programming courses~\cite[Lemma 6.1]{Rothvoss2016} is:
\begin{proposition}[{Eisenbrand, Shmonin~\cite[Thm. 1(ii)]{EisenbrandS2006}}]
\label{thm:eisenbrandshmonin}
  Any feasible and bounded ILP with $m$ constraints has a solution with support size $s\leq 2m\log(4m\Delta)$, where $\Delta$ is the largest absolute value of any entry in the constraint matrix~$A$.
\end{proposition}
Aliev et al.~\cite{AlievDEOW2018} showed the best general Eisenbrand-Shmonin-style bound $2m\log(2\sqrt{m}\Delta)$ on the support size for any ILP.
For the special cases of positively space spanning matrices~\cite{AlievADO2022}, and in the average case over all right-hand sides~\cite{OertelPW2019}, results on the order of $\mathcal{O}(m)$ have recently been obtained.
One important application of the Eisenbrand-Shmonin bound are efficient polynomial-time approximation schemes (EPTAS) for scheduling problems~\cite{Jansen2010}.
An EPTAS computes, for any problem instance~$\mathcal I$ and any $\eps > 0$, a solution whose value is within $(1+\eps)$ of the optimal solution value in time $f(1/\eps)\langle \mathcal I\rangle^{\mathcal{O}(1)}$, where~$\langle \mathcal I\rangle$ denotes the encoding size of $\mathcal I$.
Several EPTAS were devised for the classical scheduling problem of makespan minimization on uniformly related machines.
This problem is denoted as $Q||C_{\max}$ in Graham's 3-field notation~\cite{GrahamLLK1979}.
The input to $Q||C_{\max}$ is a set~$\mathcal J$ of~$N$ jobs, each of which is characterized by an integer processing time $p_j$, and a set $\mathcal M$ of $M$ machines, each of which is characterized by an integer speed~$s_i$.
The goal is to find an assignment (or schedule) $\sigma:\mathcal{J}\rightarrow\mathcal{M}$ of jobs to machines, which minimizes the maximum completion time $C_{\max}\coloneqq  \max_{i\in\mathcal M}\sum_{j\in\sigma^{-1}(i)}p_j/s_i$ of any machine.
Problem $Q||C_{\max}$ is well-known to be $\mathsf{NP}$-hard, even in the special case of unit speeds $s_1 = \ldots = s_m$, but still approximable to an arbitrary precision, in contrast to the setting of unrelated machines.
The previously fastest EPTAS for $Q||C_{\max}$ is due to Jansen, Klein and Verschae~\cite{JansenKV2016,JansenKV2020}, and runs in time $2^{\mathcal{O}\left(1/\eps\log^4\left(1/\eps\right)\right)}+N^{\mathcal{O}(1)}$.
Since their result first appeared in 2016, it has been an open question whether the exponential dependency on~$1/\eps$ can be improved.
In particular, there is still a significant gap to the best-known lower bound, which shows that a run time of $2^{\mathcal O\left((1/\eps)^{1-\delta}\right)} + N^{\mathcal O(1)}$ is not possible for any $\delta > 0$, assuming the Exponential-Time Hypothesis~\cite{ChenJZ2018}.
It is unknown, whether this gap can be improved to $\delta=0$, even under stronger assumptions such as Gap-ETH~\cite{ManurangsiR2017}.
Further, a gap remains to the best-known run time $2^{\mathcal{O}\left(1/\eps\log\left(1/\eps\right)\right)\log\left(\log\left(1/\eps\right)\right)}+N^{\mathcal{O}(1)}$ for the case of unit speeds $s_{1} = \ldots = s_{m}$~\cite{BerndtDJR2022}, which correlates with the interesting question about the gap in complexity between solving ILPs and MILPs with few constraints.

\subsection{Our Results}
\label{sec:ourresult}

\subsubsection{Combinatorial results}
Our main combinatorial result is an improvement on the fundamental Eisenbrand-Shmonin bound on the support size of any ILP with regard to a new parameter.
Namely, we replace the dependence on the largest absolute value $\Delta$ of any entry in the constraint matrix $A$ by the parameter $A_{\max}\coloneqq\max_i\|A_i\|_1$, the largest $1$-norm of any column $A_i$ in $A$:
\begin{restatable}{theorem}{NewSupportBound}
\label{thm:newsupportbound}
  Any feasible bounded ILP with $m$-row constraint matrix $A$ and largest column $1$-norm $A_{\max}$ has an optimal solution with support size $s\leq m\cdot(\log(3A_{\max})+\sqrt{\log(A_{\max})})$.
\end{restatable}
Additionally, we derive support size bounds of the form $s\leq 2m\cdot(1.46A_{\max})$.
Because all our upper bounds can equally be derived with $A_{\max}$ replaced by $\sqrt{m}\Delta$, they also improve on the constants of Aliev et. al.~\cite{AlievDEOW2018}.

For the parameter $A_{\max}$, we can show an asymptotically matching lower bound on the support size of an optimal solution:
\begin{restatable}{theorem}{NewLowerBound}
\label{thm:newlowerbound}
  For any $m\in\mathbb{Z}_{\geq0}$ and any $A_{\max}\in\mathbb{Z}_{\geq 1}$, there is an ILP with $m$ constraints, $n\coloneqq m\cdot(\lfloor\log(A_{\max})\rfloor+1)\geq m\cdot\log(A_{\max})$ variables, and maximum $1$-norm $A_{\max}$ of column vectors in the constraint matrix, whose unique optimal solution is the $\bm{1}$-vector.
\end{restatable}
To obtain this lower bound, we adapting the previously best lower bound of $m\log(\Delta)$ on the support size of an optimal solution by Berndt, Jansen and Klein~\cite{BerndtJK2021}.

\subsubsection{Algorithmic results}
We use our upper bounds on the support sizes of optimal solutions to ILPs from \autoref{thm:newsupportbound} as \emph{one} ingredient to obtain new algorithmic results.
Namely, we design a new EPTAS for $Q||C_{\max}$, which is asymptotically faster than all previous EPTAS for $Q||C_{\max}$; see \autoref{RTtable} for a survey of prior work on approximation algorithms for $Q||C_{\max}$.

\begin{restatable}{theorem}{NewQCmaxEPTAS}
\label{thm:newqcmaxeptas}
  There is an algorithm for $Q||C_{\max}$ that, for any $\eps>0$, computes a $(1+\eps)$-approx\-imate schedule for any set of $N$ jobs in time \mbox{$2^{\mathcal{O}\left(1/\eps\log^3(1/\eps)\log(\log(1/\eps))\right)}+\mathcal{O}(N)$}.
\end{restatable}
Compared to previous works, we devise a novel MILP formulation for $Q||C_{\max}$ whose constraint matrix has small column norm.
Solving this new formulation not only yields the fastest known EPTAS for $Q||C_{\max}$, but also an EPTAS which is conceptually much simpler than all previous ones.
In particular, we introduce the following simplifications in our algorithm compared to the previous EPTAS results:
\begin{itemize}
  \item In contrast to all other previous algorithms, we do not need a case distinction that splits the algorithm and its analysis into three or four different cases, depending on the processing time ratios and numbers of machines.
    Our algorithm handles all these cases simultaneously, only incurring a small increase in the constants hidden in the $\mathcal{O}$ terms.
  \item We set up an MILP where both jobs and machines with similar size and speed are combined into few distinct job and machine classes.
    All but the longest jobs and machines are then scheduled fractionally.
    Rounding the fractional allocations for relatively tiny jobs in a machine class was usually done via a separate algorithm by Lenstra, Shmoys, and Tardos~\cite{LenstraST1990}.
    Instead, we pack these jobs with a simple greedy strategy.
  \item We assign relatively huge jobs of each rounded machine speed using basic linear program properties.
    This creates a linear instead of logarithmic overhead, which we can accommodate by increasing the number of integer variables by a constant factor.
    Previous EPTAS mostly used a complex algorithm by Jansen~\cite{Jansen2012} for a particular bin packing problem.
\end{itemize}
The only remaining algorithm used as a black-box is the well-known algorithm by Lenstra~\cite{Lenstra1983} for solving MILPs with constant dimension (resp. its improvement due to Kannan~\cite{Kannan1987}).
Due to these simplifications we can easily extend our approach to further scheduling problems.

\begin{table}[htpb]\centering
    \begin{tabular*}{\textwidth}{l@{\extracolsep{\fill}} cc A}
        \toprule
            authors&year&approximation &\text{run time}&\\
        \midrule
         Gonzales, Ibarra, Sahni~\cite{GonzalezIS1977}& 1977 & $2-\frac{2}{M+1}$ & \mathcal{O}(N\log(N))&\\
         Cho, Sahni~\cite{Cho1980} & 1980 & $1+\sqrt{\frac{M-1}{2}}$ & \mathcal{O}(N\log(N))&\\
         Woeginger~\cite{Woeginger1999} & 1999 & $2-\frac{1}{M}$ & \mathcal{O}(N\log(N))&\\
         Hochbaum, Shmoys~\cite{HochbaumS1988}~ & 1988 & $1+\eps$ & N^{\mathcal{O}(1/\eps^2\log(1/\eps))}&\\
         Azar, Epstein~\cite{AzarE1998} & 1998 & $1+\eps$ & N^{\mathcal{O}(1/\eps^2)}&\\
         Jansen~\cite{Jansen2010} & 2010 & $1+\eps$ & 2^{\mathcal{O}(1/\eps^2\log^3(1/\eps))}&+N^{\mathcal{O}(1)}\\
         Jansen, Robenek~\cite{JansenR2012} & 2012 & $1+\eps$ & 2^{\mathcal{O}(1/\eps\log(1/\eps)\cdot C_{\eps})}&+N^{\mathcal{O}(1)}, \text{note\footnotemark}\\
            Jansen, Klein, Verschae~\cite{JansenKV2016,JansenKV2020}~ & 2016 & $1+\eps$ & 2^{\mathcal{O}\left(1/\eps\log^4\left(1/\eps\right)\right)}&+N^{\mathcal{O}(1)}\\
            \emph{This paper} & 2023 & $1+\eps$ & 2^{\mathcal{O}\left(1/\eps\log^3(1/\eps)\log(\log(1/\eps))\right)}&+\mathcal{O}(N)\\
        \bottomrule
    \end{tabular*}
    \caption{History of complexity results for approximating $Q||C_{\max}$.}
    \label{RTtable}
\end{table}
\footnotetext{Here, $C_{\eps}\in\mathcal{O}(1/\eps \log^{2}(1/\eps))$ denotes the proximity of the configuration ILP.}

To show the versatility of our approach, we also applied it to three different applications.

\medskip\noindent\textbf{High-Multiplicity Scheduling.}
We study the problem $Q|\mathit{HM}|C_{\max}$, where both jobs \emph{and} machines are given in the succinct high-multiplicity encoding.
That is, we are given a list $p_1,\ldots,p_d$ of job processing times and for each $p_j$ the number $n_j$ of jobs with processing time $p_j$; and a list $s_1,\ldots,s_{\tau}$ of machine speeds and for each $s_t$ the number $m_t$ of machines with speed $s_t$.
All numbers are encoded in binary, thus the input size can be exponentially smaller than in the standard encoding, which lists jobs one by one.

We are not aware of any constant-factor approximation for $Q|HM|C_{\max}$, only for the more restricted setting where only the jobs are given in a high-multiplicity encoding~\cite{FilippiRJ2009}.
We give the first approximation algorithm for this general high-multiplicity setting.

\begin{restatable}{theorem}{GreedyHMQCmax}\label{thm:hmQp}
  There is an algorithm for $Q|\mathit{HM}|C_{\max}$ that, given any instance $\mathcal{I}$ with preemptively feasible makespan~$T$, computes a schedule with makespan at most $T+p_{\max}$ in time~$\mathcal{O}(\langle\mathcal{I}\rangle)$.
\end{restatable}
Here, a preemptively feasible makespan is a feasible makespan for the preemptive relaxation,  where we allow interruption and resumption of jobs on different machines.
This directly implies a linear-time $2$-approximation for $P|\mathit{HM}|C_{\max}$ (where all machines have speed $1$).

Next, we obtain a similar result for the case of uniformly related machines:
\begin{corollary}
\label{cor:hmQ2}
  There is an algorithm for $Q|\mathit{HM}|C_{\max}$ that, for any $\eps>0$, computes a $(2+\eps)$-approximate schedule for any instance $\mathcal{I}$ in time $\mathcal{O}(\log(\eps)\cdot\langle\mathcal{I}\rangle)$.
\end{corollary}
Eventually, we show how to adapt \autoref{thm:newqcmaxeptas} to the high-multiplicity setting:
\begin{corollary}
\label{cor:hmQe}
  There is an algorithm for $Q|\mathit{HM}|C_{\max}$ that, for any $\eps>0$, computes a $(1+\eps)$-approximate schedule for any instance $\mathcal{I}$ in time \mbox{$2^{\mathcal{O}(1/\eps\log^3(1/\eps)\log(\log(1/\eps)))}+\langle\mathcal{I}\rangle^{\mathcal{O}(1)}$}.
\end{corollary}

\medskip\noindent\textbf{Few Different Machine Speeds.}
In the special case of $Q||C_{\max}$ with only $k$ distinct machine speeds, we are able to significantly improve the run time of our algorithm. 

\begin{restatable}{theorem}{SpecialCase}\label{thm:special}
  There is an algorithm for $Q||C_{\max}$ that, for any $\varepsilon > 0$, computes a $(1+\eps)$-approx\-imate schedule for any set of $N$ jobs and $k$ distinct machine speeds in time \mbox{$2^{\mathcal{O}\left(k \cdot 1/\eps \log (1/\eps) \log(\log(1/\eps)))\right)}+\mathcal{O}(N)$}.
\end{restatable}

\medskip\noindent\textbf{Few Different Uniform Machine Types.}
Jansen and Maack~\cite{JansenM2019} considered the generalization of $Q||C_{\max}$ known as $R_KQ||C_{\max}$, where each job can have up to $K$ different processing times.
Any machine belongs to one of $K$ distinct types, which determines the processing times of the jobs.
Machines of the same type can still have  different uniform speeds.
Jansen and Maack~\cite{JansenM2019} gave an EPTAS for $R_KQ||C_{\max}$ with run time
\mbox{$2^{\mathcal{O}(K\log(K)1/\eps^3\log^5(1/\eps))}+\mathcal{O}(K\cdot N)$}.
We apply our novel techniques to reduce the exponent from $\widetilde{\mathcal{O}}(K\cdot 1/\eps^3)$ down to $\widetilde{\mathcal{O}}(K\cdot 1/\eps)$:
\begin{restatable}{theorem}{Marten}
\label{thm:marten}
  There is an algorithm for $R_KQ||C_{\max}$ that, for any $\varepsilon > 0$, computes a $(1+\eps)$-approx\-imate schedule for any set of $N$ jobs in time \mbox{$2^{\mathcal{O}(K\log(K)1/\eps\log^3(1/\eps)\log(\log(1/\eps)))}+\mathcal{O}(K\cdot N)$}.
\end{restatable}

\section{Preliminaries}
\label{sec:preliminaries}
We use $\log(2)=1$, i.e., base $2$ logarithms and denote Euler's number by $e \coloneqq  \exp(1)$.
For a vector~$\bm{v}$, let $v_{\min}\coloneqq  \min_\ell v_\ell$ and $v_{\max}\coloneqq  \max_\ell v_\ell$ be its extremal entries and $\supp(\bm{v})\coloneqq  \{\ell\mid\bm{v}_{\ell}\neq 0\}$ its \emph{support}, i.e., the set of indices with non-zero entries.
For an instance $\mathcal I$, its \emph{encoding size} $\langle \mathcal I\rangle$ is given by $\langle \mathcal I\rangle\coloneqq  \sum_{x\in I}(\log(|x|+1)+1)$, the sum over the sizes in binary representations of all quantities.
For example, an instance $\mathcal{I}$ of $Q||C_{\max}$ has size $\langle\mathcal{I}\rangle$ logarithmic in the job processing times and machine speeds, but linear in the number of jobs and machines.
We generally use $i$ as index for machines, and $j$ as index for jobs. 

\medskip\noindent\textbf{Mixed-Integer Linear Programs.}
The set of feasible solutions of any MILP is
\begin{align}
  \mathcal{Q}\coloneqq  \left\{\bm{x}\in\mathbb{Z}_{\geq0}^n,\bm{y}\in\mathbb{R}_{\geq0}^r\mid\begin{pmatrix}A&B\\0&C\end{pmatrix}\cdot\begin{pmatrix}\bm{x}\\\bm{y}\end{pmatrix}=\bm{b}\right\}\tag{MILP}\label{MILP}
\end{align}
for matrices $A\in\mathbb{Z}^{m\times n}$, $B\in\mathbb{Z}^{m\times r}, C\in\mathbb{Z}^{s\times r}$ and a vector $\bm{b}\in \mathbb{Z}^{m+s}$.
The encoding size $\langle \mathit{MILP}\rangle$ of \eqref{MILP} is logarithmic in the largest coefficient $\Delta\coloneqq\max_{i,j}|A_{i,j}|$ and right-hand side, but linear in the number of variables and constraints.

We will make use of the following classical result for finding solutions of \eqref{MILP}, which was proved first by Lenstra~\cite{Lenstra1983} and obtained with improved run time by Kannan~\cite{Kannan1987}.
\begin{proposition}[Kannan~\cite{Kannan1987}]\label{MILP:Kannan}
  For any instance of \eqref{MILP}, in time $2^{\mathcal{O}(n\log(n))}\langle \mathit{MILP}\rangle^{\mathcal{O}(1)}$ one either finds a solution $(\bm{x},\bm{y})\in\mathcal{Q}$ or determines that $\mathcal{Q}=\emptyset$.
\end{proposition}
In the following, we reproduce two useful lemmata about the structure of \eqref{MILP} solutions, which are inherited from its integral and fractional parts.
\begin{lemmarep}
\label{lem:MILP:basic}
  For any instance of \eqref{MILP} and any $(\hat{\bm{x}},\hat{\bm{y}})\in\mathcal{Q}$, in time $\langle \mathit{MILP}\rangle^{\mathcal{O}(1)}$ we can find $(\hat{\bm{x}},\tilde{\bm{y}})\in\mathcal{Q}$ such that $\tilde{\bm{y}}$ is a vertex solution of the following restricted LP:
  \begin{align*}
    \left\{\bm{y}\in \mathbb{R}_{\geq0}^{r}\mid \begin{pmatrix}B\\C\end{pmatrix}\cdot\bm{y} = \bm{b}-\begin{pmatrix}A\\0\end{pmatrix}\cdot\hat{\bm{x}}\right\}\enspace .\tag{R-LP}\label{R-LP}
  \end{align*}
\end{lemmarep}
\begin{proof}
  By assumption, \eqref{R-LP} is feasible, as $\hat{\bm{y}}$ is a solution.
  With the ellipsoid algorithm~\cite[Remark 6.5.2]{GrotschelLS1988} we find a vertex solution $\tilde{\bm{y}}$ of \eqref{R-LP} in polynomial time.
\end{proof}
\begin{lemmarep}
\label{lem:MILP:support}
  For any instance of \eqref{MILP} and any $(\hat{\bm{x}},\hat{\bm{y}})\in\mathcal{Q}$ there is some $(\tilde{\bm{x}},\hat{\bm{y}})\in\mathcal{Q}$ such that $\tilde{\bm{x}}$ has minimum-cardinality support $|\supp(\bm{x})|$ of all solutions $\bm{x}$ to the restricted ILP
  \begin{align*}
    \left\{\bm{x}\in\mathbb{Z}_{\geq0}^n\mid \begin{pmatrix}A\\0\end{pmatrix}\cdot\bm{x}= \bm{b}-\begin{pmatrix}B\\C\end{pmatrix}\cdot\hat{\bm{y}}\right\}\enspace .\tag{R-ILP}\label{R-ILP}
  \end{align*}
\end{lemmarep}
\begin{proof}
  By assumption, \eqref{R-ILP} is feasible, as $\hat{\bm{x}}$ is a solution.
  Hence, it also has a solution $\tilde{\bm{x}}$ with minimum support size.
  Then $(\tilde{\bm{x}},\hat{\bm{y}})\in\mathcal{Q}$ holds because of $(\hat{\bm{x}},\hat{\bm{y}})\in\mathcal{Q}$ and $A\cdot\hat{\bm{x}}=A\cdot\tilde{\bm{x}}$.
\end{proof}
Importantly, this implies that \emph{any} support size bound for an ILP can be directly applied to the integer variables of an MILP.
One of these applications is an algorithm to solve MILPs with few constraints. 
This was presented explicitly and analyzed in terms of $m$ and $\Delta$ by Rohwedder and Verschae~\cite[p.\@ 30]{Rohwedder2019}.
As we will make extensive use of it, we formulate the underlying idea in the following lemma.
\begin{lemma}
\label{lem:MILP:support-alg}
  For any instance of \eqref{MILP} and any $s\leq n$, in time $2^{\mathcal{O}(s\log(n))}\cdot \langle \mathit{MILP}\rangle^{\mathcal{O}(1)}$ we either find a solution $\bm{x}$ of minimum-cardinality support such that $(\bm{x},\bm{y})\in \mathcal{Q}_{\leq s}\coloneqq \{(\bm{x},\bm{y})\in\mathcal{Q}\colon |\supp(\bm{x})|\leq s\}$, or determine that $\mathcal{Q}_{\leq s}=\emptyset$  %
\end{lemma}
\begin{proof}
  We exhaustively try all choices of $\supp(\bm{x})$, which are $\binom{n}{s}\leq n^s$ candidates, from $|\supp(\bm{x})| = 0$ up to $s$.
  For each choice we restrict \eqref{MILP} to $\supp(\bm{x})$, by fixing all other integer variables to $0$.
  With \autoref{MILP:Kannan} we either find a solution with $s$ integral variables, or determine the infeasibility, in time $s^{\mathcal{O}(s)}\langle \mathit{MILP}\rangle^{\mathcal{O}(1)}$.
  The claimed run time follows from that $s\leq n$.
\end{proof}
Note that a support size bound $s$ on any solution of \eqref{R-ILP} allows us to use \autoref{lem:MILP:support-alg} to find a solution with support size $s$, or decide the infeasibility of the entire \eqref{MILP}.
\autoref{lem:MILP:support-alg} directly extends to optimizing a linear objective function; and to finding a non-zero solution of minimum-cardinality support, which might be of interest for augmentation algorithms.

\medskip\noindent\textbf{Approximation Guarantees.}
Ultimately, we want to guarantee that our solution has a value that is within a factor $(1+\eps)$ of the optimum.
As we approximate in multiple places, we distribute this factor through the use of a linear auxiliary variable $\delta=\eps/c$.
The value of $c$ only depends on the algorithm, not the input, and we can use $\mathcal{O}$-techniques for $\delta\searrow0$ to ensure its existence.
I.e., adding constantly many errors of $\delta$ constantly many times still has an approximation guarantee of $(1+\mathcal{O}(\delta))^{\mathcal{O}(1)}=1+\mathcal{O}(\delta)$.

We approximate by rounding input quantities, to reduce the number of distinct machine speeds and job processing times, and by over-allocating some machines by a limited amount.
The direction of rounding is irrelevant here.
Rounding towards a harder instance (i.e., one with larger jobs and slower machines) increases the makespan.
When rounding towards a simpler instance (i.e., one with smaller jobs and faster machines), the makespan needs to be increased to maintain consistency.
Hence, either way, the makespan increases by the same asymptotic amount.

\medskip\noindent\textbf{Scheduling Problems.}
An instance $\mathcal I$ of $Q||C_{\max}$ consists of a set $\mathcal{J}$ of $N$ \emph{jobs} with processing times $p_1\leq\ldots\leq p_N$, and a set $\mathcal{M}$ of $M$ \emph{machines} with speeds $s_1\leq\ldots\leq s_M$.
In the simpler $P||C_{\max}$ setting, all machines are identical, and consequently run at unit speed, i.e., $s_1 = \ldots = s_M = 1$.

A \emph{schedule} $\sigma\colon \mathcal{J}\to \mathcal{M}$ assigns jobs to machines.
For each machine $i\in\mathcal M$ the set of jobs that are assigned to $i$ by $\sigma$ is $\sigma^{-1}(i)$.
The \emph{completion time} of a machine $i\in \mathcal{M}$ is $C_{i}^{(\sigma)} \coloneqq  \sum_{j\in \sigma^{-1}(i)}p_j/s_i$; the completion times of the machines form the vector $\bm{C}^{(\sigma)}$.
A schedule's \emph{makespan} is $C_{\max}^{(\sigma)}$, the largest completion time.
We will generally omit the schedule $\sigma$ in the notation.
The objective of our scheduling problems is to find schedules which minimize the makespan $C_{\max}$; let $\OPT(\mathcal I)=\min_{\sigma}C_{\max}^{(\sigma)}$ be the optimal makespan of instance~$\mathcal I$.
For any $\alpha\geq 1$, an $\alpha$-approximate schedule for $\mathcal I$ is a schedule of makespan at most $\alpha\cdot\OPT(\mathcal I)$.

\medskip\noindent\textbf{Run Time Analysis.}
\label{sec:RunTime}
We use $\mathcal{O}$-identities specific to the analysis of EPTAS run times $2^{\mathcal{O}(f(1/\eps))}+N^{\mathcal{O}(1)}$.
For our algorithms we have $f(1/\eps)\in\mathcal{O}(1/\eps^2)$ and can therefore use $\log(N)\in\mathcal{O}(1/\eps^2)$ when analyzing terms logarithmic in $N$, as otherwise $\mathcal{O}(N)$ dominates the run time.
The inequality $a\cdot b\leq a^2+b^2$ implies that $2^{\mathcal{O}(1/\eps)}\cdot \log^{\mathcal{O}(1)}(N) =2^{\mathcal{O}(1/\eps)}+\log^{\mathcal{O}(1)}(N)$.
Because of $\log^{\mathcal{O}(1)}(N)\subset\mathcal{O}(N)$, this gives us another way to eliminate run time factors logarithmic in $N$.
Further, as $\log(1+x)\geq x$ for $0\leq x\leq 1$ (in base 2), it holds that $\log_{1+\eps}(\cdot)=\log(\cdot)/\log(1+\eps)\in\mathcal{O(\log(\cdot)/\eps})$.

\section{Refined Support Size Bounds for Integer Linear Programs}\label{sec:supp}
In this section, we refine the general support size bounds independent of $n$ for integer programs.
Previous such bounds used as parameters the number of constraints $m$, and the largest absolute value~$\Delta$ of an entry in the constraint matrix $A$.
In our MILP formulation for $Q||C_{\max}$, we bound the support size by the \emph{maximum $1$-norm} of a column vector, denoted by $A_{\max} \coloneqq  \max_{i=1,\ldots,n}\|A_i\|_1$.
Clearly, $\Delta\leq A_{\max}\leq m\cdot\Delta$ holds, but it also means that support size bounds using only $m$ and~$\Delta$ are too coarse for some ranges.
Instead, we will show bounds using only $m$ and $A_{\max}$.

In the following, let $\mathcal{P}=\{\bm{x}\in\mathbb{R}^n_{\geq0}\mid A\bm{x}=\bm{b}\}$ for $A\in\mathbb{Z}^{m\times n}$ be a polytope and $\mathcal{P}_I=\text{Conv}(\mathcal{P}\cap\mathbb{Z}^{n})$ be the convex hull of the integer points of $\mathcal{P}$.
For simplicity, we only consider feasible ILPs $\mathcal{L}\neq\emptyset$ with $\operatorname{rank}(A)=m$ in this section.
Recall that every ILP $\mathcal{L}\coloneqq\max \{\bm{c}^{\intercal} \bm{x}\mid \bm{x}\in\mathcal{P}_I\}$ with an optimal solution also has an optimal vertex solution $\bm{v}\in\mathcal{P}_I$.
Let $S \coloneqq  \supp(\bm{v})$ be the support of the vertex solution $\bm{v}$, and $s \coloneqq  |S|$ the size of the support.

Our results are based on a bound by Aliev et.~al~\cite{AlievDEOW2018}.
We omit the inverse of the $\gcd$ of all submatrix determinants $g^{-1}$ from their result, as it is $1$ for almost all matrices.
\begin{proposition}[{Aliev et. al.~\cite[Thm. 1(2)]{AlievDEOW2018}}]
\label{prop:Supp}
  Any ILP $\mathcal{L}$ with constraint matrix $A\in\mathbb{Z}^{m\times n}$ has an optimal solution $\bm{v}\in\mathcal{L}$ with support size $|\supp(\bm{v})|=s\leq m+\log(\sqrt{\det(A\cdot A^T)})$.
\end{proposition}
In this form, the determinant of the support size bound can depend on $n$.
The strength of \autoref{prop:Supp} is the ability to restrict $A$ to the columns with non-zero variables.
For our vertex solution $\bm{v}$ with support $S=\supp(\bm{v})$, this is exactly $A_S$, the columns of the variables in the support.
Aliev et al.~\cite{AlievDEOW2018} used the inequality $\sqrt{\det(A_S\cdot A_S^T)}\leq(\sqrt{s}\Delta)^m$ to ultimately obtain the support size bound $s\leq 2m\log(2\sqrt{m}\Delta)$~\cite[Theorem 1(ii)]{AlievDOO2017}.
We analyze the term $\det(A_{S}\cdot A_{S}^{T})$ with regard to the more fine-grained parameter $A_{\max}$ to obtain a tighter bound:
\begin{lemma}
\label{lem:PDHada}
  For any matrix $A_S\in\mathbb{Z}^{m\times s}$ it holds that $\sqrt{\det(A_S\cdot A_S^T)}\leq (\sqrt{s/m}\cdot A_{\max})^m$.
\end{lemma}
\begin{proof}
  The matrix $G\coloneqq A_S\cdot A_S^T$ is symmetric and positive semi-definite. %
  If $G$ has an eigenvalue of $0$ then $\det(G)=0$ and the inequality holds.
  We will thus assume that $G$ is positive definite, i.e., all eigenvalues are positive. 
  It is a classical result that for positive definite matrices, the Hadamard inequality can be strengthened to $\det(G)\leq\prod_{i=1}^m G_{i,i}$, the product of the diagonal entries $G_{i,i}$. 
  We refer to a modern presentation by Browne et al.~\cite[Thm. 2]{BrowneEHC2021} for this fact.
  As $G=A_{S}\cdot A_{S}^{T}$, it is sufficient to bound $\varphi(A_{S})\coloneqq \prod_{i=1}^m\sum_{j=1}^s A_{S;i,j}^2$ subject to $\sum_{i=1}^m |A_{S;i,j}|\leq A_{\max}$ for $j = 1,\ldots,s$ by $(s/m\cdot A_{\max}^2)^m$ to obtain our result.
  We will first characterize a matrix $A_{S}$ such that $\varphi(A_{S})$ is maximal and then relate $\varphi(A_{S})$ to $(s/m\cdot A_{\max}^2)^m$. 
  As all entries $A_{S;i,j}$ of $A_{S}$ occur as squares or absolute values in the optimization, we can assume $A_{S;i,j}\geq 0$ in the following.
  As $\varphi$ is monotone in each variable, the matrix $A_{S}$ that maximizes $\varphi(A_{S})$ under the condition $\sum_{i=1}^m |A_{S;i,j}|\leq A_{\max}$ will fulfill these constraints with equality, i.e., $\sum_{i=1}^m |A_{S;i,j}|= A_{\max}$  for $j = 1,\ldots,s$.
  Renaming $A_{S;i,j}$ to $x_{i,j}$ thus gives us the optimization problem:
  \begin{align*}
    \max \prod_{i=1}^m\sum_{j=1}^s x_{i,j}^2\quad\text{s.t.}\;\sum_{i=1}^m x_{i,j}=A_{\max}\quad x_{i,j}\geq0\quad\for\; i = 1,\ldots,m; j = 1,\ldots,s \enspace .
  \end{align*}
  To bound the optimal solutions to this program, we first consider optimal solutions over the same region with objective function $\sum_{i=1}^m\sum_{j=1}^s x_{i,j}^2$.
  This is a convex objective function over a polyhedral region.
  Bauer's maximum principle implies that the maximum is assumed at a vertex and thus has $s$ non-zero variables~\cite{Bauer1960}.
  Consequently, we have $\sum_{i=1}^m\sum_{j=1}^s x_{i,j}^2\leq s\cdot A_{\max}^2$.
  Now, consider the problem of maximizing $\max\prod_{i=1}^m y_i$ with $\sum_{i=1}^m y_i\leq s\cdot A_{\max}^2$ and $y_i\geq0$.
  The logarithm is monotone, so we can apply it to the objective, giving $\sum_{i=1}^m\log(y_i)$ instead.
  For any solution $x^{*}_{i,j}$ maximizing $\sum_{i=1}^m\sum_{j=1}^s x_{i,j}^2$, we can compute the $y^{*}_{i}$ with $y^{*}_{i}=\sum_{j=1}^{s}{x^{*}_{i,j}}^{2}$ that will maximize $\sum_{i=1}^m\log(y_i)$.
  As the logarithm is concave, we can thus maximize  $\sum_{i=1}^m\log(y_i)$ by $y^{*}_1=\ldots=y^{*}_m=s/m\cdot A_{\max}^2$, as we could otherwise improve a solution by re-balancing it.
  Hence, $\varphi(A_{S}) \leq (s/m\cdot A_{\max}^2)^{m}$, which implies our inequality.
\end{proof}
Importantly, by substituting $A_{\max}$ with $\sqrt{m}\Delta$, our inequality in \autoref{lem:PDHada} becomes the one used by Aliev et. al.~\cite{AlievDEOW2018}.
Therefore, any of the following results can also be obtained for $\sqrt{m}\Delta$ instead of $A_{\max}$.
\autoref{prop:Supp} and \autoref{lem:PDHada} imply the essential intermediate result:
\begin{corollary}
\label{cor:SuppIneq}
  Any ILP $\mathcal{L}$ with constraint matrix $A\in\mathbb{Z}^{m\times n}$ has an optimal solution $\bm{v}\in\mathcal{L}$ with $|\supp(\bm{v})|=s$ such that $ s/m\leq 1+\log(\sqrt{s/m}\cdot A_{\max})=1+\log(A_{\max})+\log(s/m)/2$.
\end{corollary}
The proof of \autoref{prop:Supp} uses Siegel's Lemma, a deep result from transcendental number theory.
In contrast, Berndt et. al.~\cite{BerndtJK2021} derive a support size bound of the same form using only elementary methods, but with worse constants.
Building on their approach, we also derived a bound, which is slightly weaker than \autoref{cor:SuppIneq}, but requires only elementary combinatorial arguments.
\begin{lemma}
\label{lem:supp:newbound:easy}
  Any ILP $\mathcal{L}$ with constraint matrix $A\in\mathbb{Z}^{m\times n}$ has an optimal solution $\bm{v}\in\mathcal{L}$ with $|\supp(\bm{v})|=s$ such that $s/m\leq 1+\log(e)+\log(1+(s/m)\cdot A_{\max})$.
\end{lemma}

\begin{toappendix}
\begin{proof}[Proof of \autoref{lem:supp:newbound:easy}]
  The following lemmata are instrumental in bounding the support size~$s$:
\begin{lemma}[{\cite[Lemma 2]{BerndtJK2021}}]
 \label{lem:supp:bjk}
 The only vector $\bm{y}\in \{-1,0,1\}^{n}$ with $\supp(\bm{y})\subseteq S$ and $A\bm{y}=\bm{0}$ is the trivial solution $\bm{y}=\bm{0}$.
\end{lemma}
\begin{lemma}
\label{lem:supp:generalbound}
  If $\UB$ is an upper bound on the number of distinct vector values $A\bm{x}$ generated by $\bm{x}\in \{0,1\}^{n}$ with $\supp(\bm{x})\subseteq S$, then we have $s\leq\log(\UB)$. 
\end{lemma}
\begin{proof}
  We know that there are exactly $2^{s}$ such vectors $\bm{x}$.
  For $2^{s} > \UB$, the pigeonhole principle therefore implies the existence of distinct $\bm{x},\bm{x}'\in \{0,1\}^{n}$ with $\supp(\bm{x}),\supp(\bm{x}')\subseteq S$ and $A\bm{x} = A\bm{x}'$.
  However, $\bm{y} \coloneqq  \bm{x}-\bm{x}'$, which satisfies $\bm{y}\in \{-1,0,1\}^{n}$, as well as $\supp(\bm{y})\subseteq S$ and $A\bm{y}=\bm{0}$, then contradicts \autoref{lem:supp:bjk}.
  Consequently, it must hold that $2^s\leq\UB$.
\end{proof}
Eisenbrand and Shmonin~\cite{EisenbrandS2006} used the upper bound $\UB\leq(2s\cdot \Delta+1)^{m}$ and obtained the inequality $s\leq m\cdot\log(2s\cdot\Delta+1)$. 
Solving this inequality for $s$ gives a bound of $s \leq 2m\log(4m\Delta)$.
We refine this type of analysis to obtain an alternative bound in terms of $A_{\max}$ instead of $m\cdot\Delta$ in the logarithm.
\begin{lemma}
\label{lem:supp:newbound}
  We have $\UB\leq2^me^m\cdot\left(1+(s/m)\cdot A_{\max}\right)^m$.
\end{lemma}
\begin{proof}
  As $A_{\max}$ is the largest $1$-norm of any column of $A$, the triangle inequality implies that the length of any vector $A\bm{x}$ is bounded by $\|A\bm{x}\|_1=\|\sum_{i\in S}A_ix_i\|_1\leq A_{\max}\cdot s$.
  We count the number of bounded vectors, again using combinatorial techniques.
  For every bounded vector there is one non-negative vector with the same absolute value in every entry, and there are at most $2^m$ bounded vectors for every non-negative vector, by enumerating the signs of the vector entries.
  With a slack variable, any non-negative bounded vector can also be represented as a distribution of $A_{\max}\cdot s$ indistinguishable items into $m+1$ many distinguishable boxes.
  With the stars-and-bars technique \cite[Thm. 2.12]{Stanley2012} we can then derive an upper bound on the number of bounded vectors
  \begin{align*}
     \UB& \leq2^m\binom{A_{\max}\cdot s+(m+1)-1}{(m+1)-1}=2^m\binom{A_{\max}\cdot s+m}{m}
     \intertext{and by using $\binom{n}{k}\leq(e\cdot n/k)^k$ then}
     &\leq 2^m\left(\frac{e\cdot(A_{\max}\cdot s+m)}{m}\right)^m=2^m e^m\left(1+\frac{s}{m}\cdot A_{\max}\right)^m \enspace . \qedhere
  \end{align*}
\end{proof}
Taking the logarithm and rearranging the terms then proves \autoref{lem:supp:newbound:easy}.
\end{proof}
\end{toappendix}

Unfortunately, both sides of \autoref{cor:SuppIneq} are still dependent on $s$.
We resolve \autoref{cor:SuppIneq} for $s$ in two ways, both parametric in the trade-off between constant and growing terms.
In the first approach, we bound the logarithm by its tangents.
\begin{lemma}
\label{lem:supp:log}
  For any $\alpha>0$ and $x>0$, it holds $\log(x)\leq \alpha\cdot x-\log(e)+\log(\log(e)/\alpha)$.
\end{lemma}
\begin{proof}
  At $x=\log(e)/\alpha$, both sides are equal to $\log(\log(e)/\alpha)$, and the derivatives are $\alpha$.
  Hence, the affine function of the right-hand side is an upper bound on the left-hand side, as $\log(x)$ is concave.
\end{proof}
This direct approach allows us to give simple and short formulas for the bounds.
\begin{theorem}
\label{thm:support}
  For any $\alpha\in(0,1)$ there is an optimal solution $\bm{v}$ of ILP with
  \begin{align*}
    s\leq m\cdot\log(\sqrt{2\log(e)/(e\cdot\alpha)}\cdot A_{\max})/(1-\alpha)\enspace .
  \end{align*}
\end{theorem}
\begin{proof}
    Applying \autoref{lem:supp:log} with $2\alpha$ to \autoref{cor:SuppIneq} yields
    \begin{align*}
        s/m&\leq 1+\log(A_{\max})+(2\alpha\cdot s/m-\log(e)+\log(\log(e)/(2\alpha)))/2\\  %
            &\leq\log(\sqrt{2\log(e)/(e\cdot\alpha)}\cdot A_{\max})+\alpha\cdot s/m\enspace .
    \end{align*}
    We subtract $\alpha\cdot s/m$, and multiply by $m/(1-\alpha)$, which proves the claim for $0<\alpha<1$.
\end{proof}
For example, we can set $\alpha=1/2$ or $\alpha=1/11$ to obtain the bounds
$s\leq2m\cdot\log(1.46\cdot A_{\max})$ and $s\leq1.1m\cdot\log(3.42\cdot A_{\max})$.

This approach, however, only gives bounds with coefficient strictly larger than $1$ for the leading term.
To reduce this to $1$, we make use of advanced analytical function techniques, to tighter analyze the inequality of  \autoref{cor:SuppIneq}.
\begin{lemma}
\label{lem:supp:LambertW}
  For $s,m,A_{\max}\geq 1$, the inequality in \autoref{cor:SuppIneq} is equivalent to
  \begin{align*}
    s/m\leq-\log(e)\cdot\mathcal{W}_{-1}(-1/(2\log(e)A_{\max}^2))/2,
  \end{align*}
  where $\mathcal{W}_{-1}$ is the $-1$ branch of the LambertW-function, the inverse function of~$x\mapsto xe^x$.
\end{lemma}
\begin{proof}
  We substitute $s/m$ by $-\log(e)\cdot y/2$, and rearrange to obtain:
  \begin{align*}
    -\log(e)\cdot y/2&\leq 1+\log(A_{\max})+\log(-\log(e)\cdot y/2)/2\\
    \log(-\log(e)\cdot y)+\log(e)\cdot y&\geq -(2+2\log(A_{\max})-1)=-1-2\log(A_{\max})\\
    -\log(e)\cdot y\cdot 2^{\log(e)\cdot y}&\geq 1/(2\cdot A_{\max}^2)\Leftrightarrow
    y\cdot e^y\leq-1/(2\log(e)\cdot A_{\max}^2).
  \end{align*}
  For the right-hand side $z\coloneqq -1/(2\log(e)\cdot A_{\max}^2)$, we have $-1/e\leq-1/(2\log(e))\leq z\leq 0$.
  Therefore, the inequality is satisfied exactly when $\mathcal{W}_{-1}(z)\leq y\leq \mathcal{W}_0(z)$ holds, where $\mathcal{W}_k$ are the real branches of the aforementioned LambertW-function.
  
  Next, we show $y\leq \mathcal{W}_0(z)$ does not restrict any relevant values. 
  For $x\in[{-1/e},0]$, we have $W_0(x)\geq e\cdot x$.
  Furthermore $s/m\cdot A_{\max}^2\geq 1\geq e/4$ holds for the relevant values of $s/m\geq1$ and $A_{\max}\geq 1$.
  Therefore $\mathcal{W}_0(z)\geq -e/(2\log(e)A_{\max}^2)\geq-2\cdot (s/m)/\log(e)=y$.
  This shows that the positive solutions are only constrained from above, by $\mathcal{W}_{-1}(z)$.
  Applying the resubstitution of $y$ to its lower bound gives the claimed result.
\end{proof}
Now, we apply techniques analogous to Chatzigeorgiou~\cite{Chatzigeorgiou2013} to prove parametric bounds for the $\mathcal{W}_{-1}(z)$ branch, which we optimize for $z\rightarrow0$ instead of $z\rightarrow -1/e$.
\begin{lemma}
\label{bound:LambertW}
  For any $\alpha>0$ and $u\geq0$, it holds $-\mathcal{W}_{-1}(-e^{-u-1})\leq u+\sqrt{2\alpha\cdot u}+\alpha-\ln(\alpha)$.
\end{lemma}
\begin{proof}
  Chatzigeorgiou~\cite{Chatzigeorgiou2013} showed for $x=-\mathcal{W}_{-1}(-e^{-u-1})-1$ that $g(x) = u$, where $g(x)\coloneqq x-\ln(1+x)$.
  To show our result, it thus suffices to show that for all $x\in\mathbb{R}>0$ and some additive term $\beta$, only dependent on $\alpha$, it holds that
  \begin{align*}
    -\mathcal{W}_{-1}(-e^{-u-1})=x+1\leq g(x)+\sqrt{2\alpha\cdot g(x)}+\beta+1=u+\sqrt{2\alpha\cdot u}+\beta+1\enspace.
  \end{align*} 
  By definition of $g$, this reduces to showing $f(x)\coloneqq -\ln(1+x)+\sqrt{2\alpha\cdot(x-\ln(1+x))}+\beta\geq0$.
  The function~$f$ has a unique minimum, since we will show that $f'(x)=0$ has only one solution and $\lim_{x\rightarrow\infty}f(x)=\infty$.
  Consider the critical condition:
  \begin{align}\label{LambertW:criteq}
  \frac{\mathrm d}{\mathrm dx}f(x)=\frac{-1}{1+x}+\frac{(1-\frac{1}{1+x})\cdot\alpha}{\sqrt{2\alpha(x-\ln(1+x))}}=0\quad\Leftrightarrow\quad \frac{x\alpha}{\sqrt{2\alpha(x-\ln(1+x))}}=1\enspace.
  \end{align}
  \autoref{LambertW:criteq} can not be solved for $x$ directly.
  Instead, we show that there can be only one solution, a global minimum, because the second derivative of $f(x)$ is always positive, making $f(x)$ monotonously increasing.
  \begin{multline*}
  \frac{\mathrm d}{\mathrm dx^2}f(x)=\frac{\mathrm d}{\mathrm dx}\left(\frac{\alpha x}{\sqrt{2\alpha(x-\ln(1+x))}}-1\right)=\frac{-x(1-\frac{1}{1+x})\alpha^2}{(2\alpha(x-\ln(1+x)))^{3/2}}+\frac{\alpha}{\sqrt{2\alpha(x-\ln(1+x))}}>0\quad\Leftrightarrow\\
  \frac{-\alpha^2x^2/(1+x)}{2\alpha(x-\ln(1+x))}+\alpha>0\quad\Leftrightarrow\quad\frac{x^2}{1+x}<2(x-\ln(1+x))\quad\Leftrightarrow\quad \frac{x(2+x)}{2(1+x)}\geq\ln(1+x)
  \end{multline*}
  We know \autoref{LambertW:criteq} must hold at the global minimum.
  Hence, substituting it in the inequality $f(x)\geq0$ and applying \autoref{lem:supp:log} on $\ln(1+x)$ bounds the minimal value of $f(x)$ by 
  \begin{equation*}
  f(x)=-\ln(1+x)+\alpha x+\beta\geq-\alpha(1+x)-\ln(1/\alpha)+1+\alpha x+\beta\geq0\quad\Leftrightarrow\quad \beta\geq \alpha+\ln(1/\alpha)-1\enspace. 
  \end{equation*}
  We conclude that $-\mathcal{W}_{-1}(-e^{-u-1})\leq u+\sqrt{2\alpha\cdot u}+\alpha+\ln(1/\alpha)$, which is the claim.
\end{proof}
From \autoref{lem:supp:LambertW} and \autoref{bound:LambertW} we derive our asymptotically tight support size bound.
\begin{theorem}
\label{thm:LambertW}
  For any $\alpha'>0$ there is an optimal solution $\bm{v}$ of ILP with
  \begin{align*}
    s\leq m\cdot(\log(A_{\max})+\sqrt{\alpha'(\log(A_{\max})+0.05)}+\alpha'/2+\log(\sqrt{1/\alpha'})+1.03)\enspace.
  \end{align*}
\end{theorem}
\begin{proof}
  We need to rewrite the argument $z\coloneqq -1/(2\log(e)A_{\max}^2)$ in \autoref{lem:supp:LambertW} to the form $-e^{-u-1}$ used in \autoref{bound:LambertW}.
  Hence, solving $z=-e^{-u-1}$ gives $u=\ln(2\log(e)A_{\max}^2/e)=2\ln(A_{\max})+\ln(2\log(e)/e)$.
  We now substitute $\alpha$ by $\alpha'/\log(e)$ to get $\log$ instead of $\ln$, and through calculation obtain the bound:
  \begin{align*}
      s/m&\leq-\log(e)\mathcal{W}_{-1}(z)/2\leq\log(e)(u+\sqrt{2\alpha\cdot u}+\alpha-\ln(\alpha))/2\\
      &\leq\log(e)(u+\sqrt{2\alpha'/\log(e)\cdot u}+\alpha'/\log(e)-\ln(\alpha'/\log(e)))/2\\
      &\leq\log(e)u/2+\sqrt{\log(e)\alpha'u/2}+\alpha'/2-\log(\alpha'/\log(e))/2\\
      &\leq\log(A_{\max})+\log(2\log(e)/e)/2+\sqrt{\log(e)\alpha'u/2}-\log(\alpha'/\log(e))/2+\alpha'/2\\
      &\leq\log(A_{\max})+\sqrt{\alpha'(\log(A_{\max})+\log(2\log(e)/e)/2)}+\alpha'/2+\log(\log(e)\sqrt{2/(e\alpha')})\enspace.
  \end{align*}
  Inserting numerical values for terms independent of $\alpha'$ and $A_{\max}$ gives the desired result.
\end{proof}
For $\alpha=1$ and $A_{\max}\geq1$ we obtain the particularly simple bounds
\begin{equation*}
  s\leq m\cdot(\log(A_{\max})+\sqrt{\log(A_{\max})+0.05}+1.53)\leq m\cdot(\log(3A_{\max})+\sqrt{\log(A_{\max})}) \enspace .
\end{equation*}
This immediately implies our main support size bound:
\NewSupportBound*
The results of this section directly extend to other $p$-norm bounds on the column vectors of $A$:
\begin{corollaryrep}
  Let $A_{\max}^{(p)}\coloneqq \max_i\|A_i\|_p$ be the largest $p$-norm of columns of $A$ for $p \geq 1$. 
  It holds that
  \begin{align*}
      \log(A_{\max})\leq \log(m^{1-1/p}A_{\max}^{(p)})\quad\text{and}\quad \log(A_{\max})\leq p\cdot\log(A_{\max}^{(p)})\enspace.
  \end{align*}
\end{corollaryrep}
\begin{proof}
  H{\"o}lder's inequality implies that $\|\bm{x}\|_1\leq d^{1-1/p}\|\bm{x}\|_p$ for any vector $x\in\mathbb{Z}^d$.
  The first inequality follows from $d\leq m$.
  As $x$ is integral, we can also bound the number of non-zero entries as $d\leq\|\bm{x}\|_p^p$.
  The second inequality then follows from $(\|\bm{x}\|_p^p)^{1-1/p}\|\bm{x}\|_p=\|\bm{x}\|_p^p$.
\end{proof}
In order to understand how tight our bound is, we adapt a construction by Berndt et. al.~\cite{BerndtJK2021} to obtain a lower bound on the support size and close the asymptotic gap.
\NewLowerBound*
\begin{proof}
  With $d\coloneqq  \lfloor\log(A_{\max})\rfloor$ we construct an ILP as follows:
  \begin{align*}
  \max&\;\,\begin{pmatrix}3^0&\cdots&3^d&3^0&\cdots&3^{d}&\cdots&3^0&\cdots&3^d\end{pmatrix}\\
  \text{s.t.}&\begin{pmatrix}
      2^0 & \cdots & 2^d & 0 & & & \cdots & & & 0 \\
      0 & \cdots & 0 & 2^0 & \cdots & 2^d & \cdots & 0 & \cdots & 0 \\
          & \vdots & & & & & & \ddots &\\
      0 & & & \cdots & & & 0 & 2^0 & \cdots & 2^d
  \end{pmatrix}=\begin{pmatrix}2^{d+1}-1\\2^{d+1}-1\\\vdots\\2^{d+1}-1\end{pmatrix} \enspace .
  \end{align*}
  Because of $\sum_{i=0}^d2^i=2^{d+1}-1$, the $\bm{1}$-vector is a solution.
  All coefficients are positive, hence in any solution the value of the variables with coefficient $2^d$ must be less than $2$.
  For any other variable $x_1$, if it is $x_i\geq2$, we can increase the objective by setting $x_i\coloneqq x_i-2;x_{i+1}\coloneqq x_{i+1}+1$.
  Since all objective coefficients are positive, the $\bm{1}$-vector is the unique optimal solution.
\end{proof}
Hence, \autoref{thm:LambertW} is exact in the dominant term.
We pose the question, whether there is a support size bound of the form $m\cdot\log(c\cdot A_{\max})$ for some constant $c$, as an interesting open problem.

\section{An Efficient Approximation Scheme for Makespan Minimization on Uniformly Related machines}\label{sec:Approx_algorithm}
Our algorithm follows a typical structure of approximation schemes. %
First, we \emph{preprocess} the input, by discarding jobs and machines which are so short or slow that assigning them naively is acceptable.
Next, we perform a \emph{binary search} on the makespan, to reduce the optimization problem to a feasibility problem.
Then, we \emph{round} the remaining processing times and machine speeds, according to our makespan guess, to make the resulting instance more structured.
Now, we \emph{construct} an MILP, whose feasibility is equivalent to the existence of a schedule.
Finally, we solve the MILP and \emph{transform} the solution into a schedule.

\subsection{Preprocessing}
\label{sec:pre}
We reduce the number of parameters bounding the instance to the number of jobs $N$ and a constant fraction $\delta$ of the approximation guarantee $\eps$.
To enforce $N \geq M$ we potentially drop the $M-N$ slowest machines, as there is an optimal solution not assigning them a job.
\paragraph*{Step 1: Removing Negligible Machines and Jobs}
\label{QCmax:para:Step1}
We remove all machines slower than $\delta\cdot s_{\max}/N$ and all jobs shorter than $\delta\cdot p_{\max}/N$.
Compensating for the lost processing times on a longest job and the machine speeds on a fastest machine introduces an error of at most a factor $(1+\delta)$ respectively.
Now we have $p_{\min}>\delta\cdot p_{\max}/N$ and $s_{\min}>\delta\cdot s_{\max}/N$ and the largest ratios of job processing times $p_{\max}/p_{\min}<N/\delta$ and machine speeds $s_{\max}/s_{\min}<N/\delta$ are bounded only by the parameters $N$ and $\delta$.
\paragraph*{Step 2: Preround the Inputs}
\label{QCmax:para:Step2}
To achieve a linear run time, we preround the machine speeds and processing times to fewer distinct values.
These are rerounded again more carefully at every iteration of the binary search, reducing the run time at the cost of a limited accuracy loss.
We round every processing time $p_j$ and machine speed $s_i$ down to the next power of $(1+\delta)$, introducing errors of no more than $(1+\delta)$ by construction.
Let $\tilde{\eta}_j$ be the number of jobs with rounded processing times $\tilde{p}_j$ be and $\tilde{\mu}_i$ the number of machines with rounded speed $\tilde{s}_j$.
Due to step~1, the amount of \emph{distinct} values after rounding is bounded by $\log_{1+\delta}(p_{\max}/(p_{\max}\delta/N))=\log_{1+\delta}(s_{\max}/(s_{\max}\delta/N))\in\mathcal{O}(1/\delta \log(1/\delta \cdot N))$.
Because our inputs are sorted, the rounding above can be performed in time $\mathcal{O}(N+1/\delta \log(1/\delta \cdot N))$.
\paragraph*{Step 3: Binary Search for the Makespan}
\label{QCmax:para:Step3}
We reduce finding the optimal makespan $\OPT(\mathcal I)$ to successively checking whether a schedule with makespan $T$ is realizable.
The processing time of a longest job on a fastest machine is a lower bound: $\OPT(\mathcal I)\geq p_{\max}/s_{\max}$.
The schedule assigning all jobs to a fastest machine is an upper bound $\OPT(\mathcal I)\leq\sum_{j=1}^N p_i/s_{\max}\leq N p_{\max}/s_{\max}$.
Hence, we can use a binary search for $\OPT(\mathcal{I})$ in the interval $[p_{\max}/s_{\max},N\cdot p_{\max}/s_{\max}]$ of ratio $N$.
As accuracy up to a factor $(1+\delta)$ is sufficient, we only need to consider integer powers of $(1+\delta)$.
Our binary search therefore adds a factor of $\mathcal{O}(\log_{1+\delta}(N))=\mathcal{O}(1/\delta \log(N))$ to the run time of the following steps.
We denote the current makespan in the binary search by $T$.
This transforms the problem into either finding a schedule with makespan $(1+\mathcal{O}(\delta))\cdot T$, or deciding that no schedule with makespan~$T$ exists.
\paragraph*{Step 4: Rounding Machine Speeds and Job Processing Times}
\label{QCmax:para:Step4}
A $(1+\eps)$-approximate schedule~$\sigma$ satisfies, for each machine $i$, the equivalent inequalities 
\begin{equation} \label{eq:unitT}
  \sum_{j\in\sigma^{-1}(i)}\frac{p_j}{s_i}\leq (1+\eps)\cdot T\quad \Leftrightarrow \quad \sum_{j\in\sigma^{-1}(i)}\frac{p_j}{T}\leq (1+\eps)\cdot s_i\enspace .
\end{equation}
With the aforementioned bounds on the makespan we have enforced the descending chain
\begin{align*}
  s_{\max}\geq p_{\max}/T \geq p_{\min}/T> \delta\cdot p_{\max}/(N\cdot T)\geq\delta s_{\max}/N^2\enspace .
\end{align*}
Therefore, all quantities are in the interval $I\coloneqq (\delta\cdot s_{\max}/N^2,s_{\max}]$ of ratio $N^2/\delta$,  especially the scaled processing times $p_j/T$.
See also \autoref{fig:range} for an overview on the relations of the parameters.
We now scale each processing time $p_j$ with $1/T$; this yields an instance with scaled processing times
$\tilde p_j = 1/T \cdot p_j$, equivalent to our original instance by \autoref{eq:unitT}.
\begin{figure}
  \centering
  \scalebox{1}{
    \input{range.tikz}
  }
  \caption{Overview on the range of parameters.\label{fig:range}}
\end{figure}
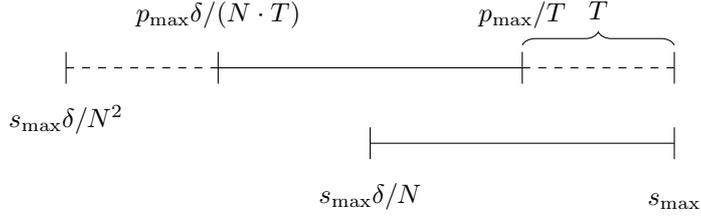
The crucial innovation by Berndt, Deppert, Jansen and Rohwedder~\cite{BerndtDJR2022} was to combine exponential and linear rounding to obtain both sufficient accuracy with few values and useful structural properties.
We adapt their approach to cover our interval $I$.
With $\kappa\coloneqq \lceil\log(1/\delta \cdot N^2)\rceil$, consider the points $b_{k,0}\coloneqq s_{\max}\cdot 2^{-k}$ for $k=0,\ldots,\kappa$.
The intervals $[b_{k+1,0},b_{k,0}]$ become exponentially finer, but have a ratio of $2$, not the necessary $(1+\delta)$.
Therefore, with $\lambda\coloneqq \lceil 1/\delta\rceil$, we add the points $b_{k,\ell}\coloneqq (1-\ell/(2\lambda))\cdot b_{k,0}$ for $k=0,\ldots,\kappa-1$ and $\ell=0,\ldots,\lambda$, which is a linear interpolation between $b_{k,0}$ and $b_{k+1,0}$.
Simple calculations show the relations $b_{k,\ell-1}/b_{k,\ell}=1+1/(2\lambda-\ell)\leq 1+\delta$ and $b_{k+1,0}=b_{k,0}/2=b_{k,\lceil 1/\delta\rceil}$ for all $k=0,\ldots,\kappa-1$ and $\ell=1,\ldots,\lambda$, which means we have achieved the necessary precision.
Crucially, two jobs within a linear interval of the same parity combine exactly to a job in the next larger linear interval.
Formally, this is described by
\begin{align}
  b_{k,\ell}+b_{k,\ell'}=\left(2-\frac{\ell+\ell'}{2\lambda}\right)\cdot b_{k,0}=\left(1-\frac{\frac{\ell+\ell'}{2}}{2\lambda}\right)\cdot b_{k-1,0}=b_{k-1,\frac{\ell+\ell'}{2}}\label{eq:conf:pairs}
\end{align}
for all $0\leq\ell\leq\ell'\leq\lambda$ with $\ell+\ell'$ divisible by $2$, i.e., $\ell$ and $\ell'$ are both odd or both even.
\autoref{eq:conf:pairs} will allow us to significantly reduce the number of configurations and the maximum $1$-norm of a column.
To simplify our notations, we re-index the values of $b_{k,\ell}$ in descending order by setting $b_{r}\coloneqq b_{k(r),\ell(r)}$ with $k(r)=\lfloor (r-1)/\lambda\rfloor$ and $\ell(r)=(r-1)\bmod\lambda$, such that $b_{k\cdot\lambda+\ell+1}=b_{k,\ell}$.
Therefore, the interval~$I$ is covered completely by the $\tau\coloneqq \kappa\cdot\lambda\in\mathcal{O}(1/\delta \log(1/\delta \cdot N))$ intervals $(b_{r+1},b_r]$ for $r=1,\ldots,\tau$.
Finally, we round every machine speed $\tilde{s}_i$ as well as every scaled processing time $\tilde{p}_j$ in $(b_{r+1},b_r]$ up to $b_r$. %
Let $\mu_i$ and~$\eta_j$ be the number of machines with scaled speed $b_i$ and jobs with scaled processing time $b_j$ respectively.
This takes time $\mathcal{O}(1/\delta\log(1/\delta \cdot N))$, which means steps 1 to 4 can be performed in total time 
\begin{equation*}
\mathcal{O}(N+1/\delta \log(N) (1/\delta \log(1/\delta \cdot N)+\tau))=\mathcal{O}(N+1/\delta^2 \log^2(1/\delta \cdot N))\subseteq \mathcal{O}(1/\delta^{2+2\cdot 2})+\mathcal{O}(N)\enspace .
\end{equation*}
\autoref{fig:group} contains a graphical representation of the exponential intervals $[b_{k+1,0},b_{k,0}]$ and the linear sub-intervals $[b_{k,i+1},b_{k,i}]$.
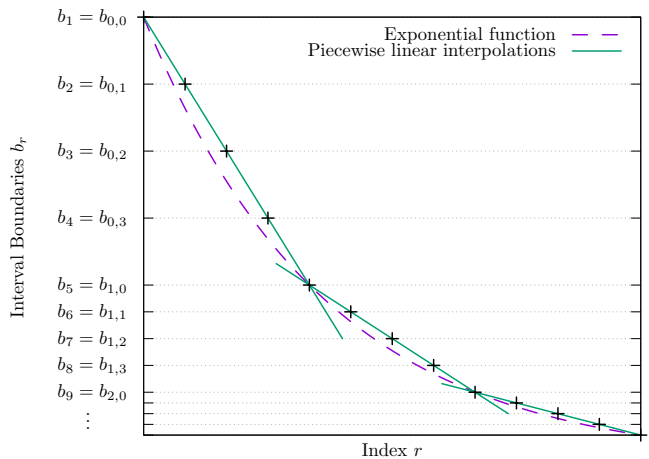
\begin{figure}[h]
  \centering
  \scalebox{.7}{
    \input{rounding.tikz}
  }
  \caption{Sketch of the grouping of the interval $I$ into sub-intervals.\label{fig:group}}
\end{figure}

\subsection{Solving an MILP Formulation}
\label{sec:QCmaxMILP}
Consider the resulting instance after all steps from \autoref{sec:pre} have been applied. 
Nearly all previous approaches used a mix of \emph{configuration} variables that determine the complete schedule of a machine and
\emph{assignment} variables that determine the position of a single job.
We now combine these different variables into a unified structure called \emph{recursive} configurations.
The core idea of our formulation is to use an additional \emph{virtual} machine $i$ of speed~$b_i$ by placing a \emph{corresponding} job of same size $b_i$ on a faster machine $i'$ with $b_{i'} <  b_{i}$, in effect simulating the virtual machine $i$ on machine $i'$.
In other words, by placing more jobs than the problem requires, we are also allowed to use more machines of the same size than the problem provides.
By applying this idea recursively, we can cover a large range of job processing times with configurations of limited range only, as the virtual machines allow us to merge several short jobs (with respect to to a certain machine speed) into a long job. 
Furthermore, this approach also allows us to successively build a configuration from other configurations, significantly reducing the number of overall configurations necessary.

For $k=0,\ldots,\kappa$, we define the set $G_k \coloneqq \{r\in \{1,\ldots,\tau\} \mid b_r \in [b_{k,0},b_{k,\lambda}]\}$ of indices of affine slices in our rounding scheme and partition it into the sets $G_k^{\text{even}}\coloneqq \{r\in G_k\mid r=0\bmod 2\}$ and $G_k^{\text{odd}}\coloneqq G_k\backslash G_k^{\text{even}}$.
For $i = 1,\ldots,\tau$, we define the set $H_i\coloneqq \{j\in \{1,\ldots,\tau\}\mid b_j \in (\delta b_i,b_i]\}$ of indices of \emph{long} job speeds, which use more than $\delta b_i$ and less than the entire machine speed $b_i$.

We will now define \emph{configurations}.
All configurations are vectors $\bm{\gamma}$ with $\tau$ entries, each representing multiples of scaled processing times in $\bm{b}=(b_{1},\ldots,b_{\tau})$.
In the following, we describe the configurations for machines with speed $b_{i}$.

The set $\mathcal{C}_{i}^{(1)}$ contains all the exact combinations $b_j+b_{j'}=b_i$ as described in \autoref{eq:conf:pairs}.
With $\bm{e}_j\in\{0,1\}^{\tau}$ being the $j$-th unit vector, let
\begin{align*}
  \mathcal{C}_i^{(1)} \coloneqq  \{\bm{e}_{j}+\bm{e}_{j'} \mid j,j'\in G_{k(i)-1} \text{ and } j+j'=2\cdot(i-\lambda) \}\enspace .
\end{align*}
The second set $\mathcal{C}_{i}^{(2)}$ contains the remaining feasible configurations of long jobs, with at most one job at even and odd positions in an affine slice $G_k$:
\begin{align*}
  \mathcal{C}_{i}^{(2)} \coloneqq  \{\bm{\gamma}\in\{0,1\}^\tau\mid{}
  &\bm{\gamma}\cdot\bm{b}\leq b_i\text{ and }\supp(\bm{\gamma}) \subseteq H_i\text{ and for all}\\
  &k=0,\ldots,\kappa-1:\sum_{r \in G_k^{\text{even}}} \gamma_{r} \leq 1 \text{ and } \sum_{r\in G_k^{\text{odd}}}\gamma_{r} \leq 1\}\enspace .
\end{align*}
Finally, the set of configurations $\mathcal{C}_{i}$ for machines with speed $b_i$  is defined as $\mathcal{C}_{i}\coloneqq \mathcal{C}_{i}^{(1)}\cup\mathcal{C}_{i}^{(2)}$.
The total number of entries in a configuration $\bm{\gamma}$ is at most $\|\bm{\gamma}\|_1\leq 2\log(2 \cdot 1/\delta)$.
Only long jobs from $H_i$ are used, and for each~$k$ there is at most one job at an even or odd position in~$G_k$, respectively. 
We can bound the number of configurations in $\mathcal{C}_{i}^{(1)}$ by $\lambda^2$, and the number of configurations in $\mathcal{C}_{i}^{(2)}$ by $(\lambda^2)^{2\log(2\cdot 1/\delta)}\in 2^{\mathcal{O}(\log^2(1/\delta))}$, which implies $|\mathcal{C}_i|\in 2^{\mathcal{O}(\log^2(1/\delta))}$.

We require integrality in the configuration variables only for the fastest $L\coloneqq \lambda\lceil\log(1/\delta^3\log(1/\delta))\rceil \in\mathcal{O}(1/\delta\log(1/\delta))$ machine speeds.
Intuitively, we just need to assign configurations integrally on these machines, as all remaining configurations are very short relative to the machines with the fastest speed.
Hence, we can assign them fractionally and round them later on.
The overhead from rounding is then scheduled on a fastest machine.

Recall that the number of jobs with processing time $b_j$ in configuration $\bm{\gamma}$ is $\gamma_j$, the number of machines with speed $b_i$ is $\mu_i$, and the number of jobs with scaled processing time~$b_j$ is $\eta_{j}$. 
The resulting MILP is:
\begin{align*}
  \sum_{\bm{\gamma}\in\mathcal{C}_i}x_{i,\bm{\gamma}}-\mu_i&\overset{\oplabel{eq:virtual}}{=}\sum_{i'=1}^{\tau}\sum_{\bm{\gamma}\in\mathcal{C}_{i'}}\gamma_i\cdot x_{i',\bm{\gamma}}-\eta_i\overset{\oplabel{eq:pack}}\geq 0\qquad&\for\;i =1,\ldots,\tau\\
  x_{i,\bm{\gamma}}&\geq0\;&\for\;i = 1,\ldots,\tau;~\bm{\gamma}\in\mathcal{C}_i&\nonumber\\
  x_{i,\bm{\gamma}}& \in\mathbb{Z}_{\geq 0}\;&\for\;i =1,\ldots, L;~\bm{\gamma}\in\mathcal{C}_i .&\tag{recursive-MILP}\label{recursive-MILP}
\end{align*}
The constraints~\eqref{eq:virtual} enforce that the number of additional virtual machines of any speed equals the number of additional corresponding jobs of that same size scheduled somewhere else.
The constraints~\eqref{eq:pack} ensure that at least as many jobs of each size are assigned in configurations, as are required by the problem.
We show \eqref{recursive-MILP} is feasible, up to an approximation factor, for a feasible instance.
\begin{lemmarep}
\label{MILPround1}
  If the original $Q||C_{\max}$-instance $\mathcal{I}$ has a schedule $\sigma$ with makespan $T$, then \eqref{recursive-MILP} is feasible for makespan $(1+17\delta)\cdot T$.
\end{lemmarep}
\begin{proof}
  The schedule $\sigma$ specifies which jobs are scheduled on any machine $i$ and that those have sufficient speed.
  We perform steps 1--4, which uses additional speed of at most a factor $(1+\delta)^7$.
  If a single job of processing time $b_i$ has been assigned to a machine of speed $b_i$ by $\sigma$, then we create a new configuration which assigns this job on that machine, and speed up the machine by a factor $(1+\delta)$.
  Otherwise, the difficulty is finding configurations for jobs which are not long, i.e., less than $\delta\cdot b_i$.
  We repeatedly combine pairs of jobs according to \autoref{eq:conf:pairs} until we have at most one job at an even and odd position in each affine slice $G_k$ on a machine.
  We partition these jobs $j$ by their value of $\ell=\lfloor\log_{\delta}(b_j/b_i)\rfloor$, i.e., into slices with ratio~$\delta$.
  The total processing time of jobs in slice $\ell$ is bounded by $b_i\cdot\delta^{\ell}\cdot2\sum_{j=0}^{\infty}2^{-j}=4\cdot b_{i}\cdot \delta^\ell$.
  We iteratively bundle all jobs from the same slice, starting with the last slice until only the first slice is left, into at most $8$ configurations of size $b_i\cdot\delta^\ell$.
  Each configuration can be packed at least half full by the greedy algorithm.
  The resulting configurations for slice $\ell$ might have to be rounded in size to the next $b_r$.
  Hence, the additional speed introduced for slice~$\ell$ is bounded by $8\cdot b_i\cdot\delta^{\ell}\cdot(1+\delta)$.
  After the creation of such a configuration, if possible, it gets combined again, which uses no additional speed.
  Eventually, only jobs in the slice for $\ell=0$ are left.
  These can be scheduled exactly in a configuration on the machine, by the premise.
  The additional load incurred on a machine can be bounded by $\sum_{\ell=1}^{\infty}8\cdot b_{i}\cdot \delta^{\ell}\cdot (1+\delta)=8\cdot (1+\delta)/(1-\delta)\cdot\delta \cdot b_i$.
  For $\delta\leq1/100$ this is less than $9\cdot \delta\cdot b_i$ and in total a factor $(1+\delta)^7(1+9\delta)\leq (1+17\delta)$ on the makespan, as claimed.
  The values $\bm{x}$ derived from this satisfy constraints~\eqref{eq:virtual} and~\eqref{eq:pack} by construction.
\end{proof}
For \eqref{recursive-MILP}, the integer subproblem has $m=2L\in\mathcal{O}(1/\delta\log(1/\delta))$ constraints and maximal $1$-norm of a column $\max\|A_i\|_1\in\mathcal{O}(\log(1/\delta))$.
The first $L$ machine speeds have $L\cdot 2^{\mathcal{O}(\log^2(1/\delta)} = 2^{\mathcal{O}(\log^2(1/\delta))}$ configurations, i.e., integer variables.
By \autoref{thm:newsupportbound} and \autoref{lem:MILP:support}, if there is a feasible solution of \eqref{recursive-MILP}, then there is also one with $\mathcal{O}(1/\delta\log(1/\delta)\log(\log(1/\delta)))$ positive integer variables.
\autoref{lem:MILP:support-alg} thus implies that we can solve  \eqref{recursive-MILP}($S$) in time $2^{\mathcal{O}(1/\delta\log^3(1/\delta)\log(\log(1/\delta)))}\cdot \log(N)^{\mathcal{O}(1)}$, as the encoding size is $\langle \text{\eqref{recursive-MILP}}(S)\rangle \leq (1/\delta \log(N))^{\mathcal{O}(1)}$.
Note that this is the dominant run time in terms of $1/\delta$.
We thus either find a feasible solution for the current makespan guess $T$, or discover that no such solution exists.
In the latter case, we discard our current makespan guess and increase it in the next step of the binary search.

\subsection{Constructing a Schedule}
\label{sec:round}
We need to construct a schedule from a solution $\bm{x}^\star$ to \eqref{recursive-MILP}.
By constraints~\eqref{eq:pack}, we know that $\bm{x}^{\star}$ schedules each job in some configuration.
By constraints~\eqref{eq:virtual}, the number of virtual machines equals the additional number of corresponding jobs with equal size scheduled in some configuration.
Managing to assign all configurations within an approximate makespan to machines would therefore be a valid schedule.
We use the \autoref{algoround} for this:
\begin{printalgo}
\pcb[head=\parbox{\textwidth-8.8pt}{\textsf{AssignConfsToMachines\\
             Input: A feasible solution $\bm{x}^\star$ to \eqref{recursive-MILP}.\\
             Output: An approximate schedule $\sigma$ to the pre-processed instance.}}]
             {\textsf{AssignConfsToMachines}}{\label{algoround}
    \t \pcfor \text{decreasing machine speeds $b_{i}$}:\\
    \t[2] \pcfor \text{each $\bm{\gamma}\in \mathcal{C}_{i}$}:\\
    \t[3] \text{assign $\lfloor x_{i,\bm{\gamma}}\rfloor$ copies of $\bm{\gamma}$ to machines with speed $b_{i}$}\\
    \t[3] \pcif x_{i,\bm{\gamma}}\not\in \mathbb{Z}_{\geq 0}: \text{assign another copy of $\bm{\gamma}$ to a fastest machine}\\
    \t[3] \pcfor \text{each processing time $b_{j}$ used in $\bm{\gamma}$}:\\
    \t[4] \text{assign as many jobs of processing time $b_{j}$ to machines with configuration $\bm{\gamma}$}\\
    \t[4] \pclinecomment{stop when all jobs are packed or all configurations are filled}\\
    \t[4] \pcfor \text{each job in a configuration $\bm{\gamma}$ not filled}:\\
    \t[5] \text{create a virtual machine with the same speed as the corresponding job}\\
    \t \pcreturn \text{the resulting schedule $\sigma$}
}
\end{printalgo}
\begin{lemmarep}
\label{lem:recMILP:schedule}
  \autoref{algoround} gives a schedule of makespan at most $(1+5\delta)\cdot T$ from a feasible solution $\bm{x}^{\star}$ to \eqref{recursive-MILP} in  time \mbox{$2^{\mathcal{O}(1/\delta)}+\mathcal{O}(N\log^2(N))$}.
\end{lemmarep}
\begin{proof}
  The algorithm assigns $\lfloor x_{i,\bm{\gamma}}\rfloor + 1 \geq x_{i,\bm{\gamma}}$ configurations, at least as many as~$\bm{x}^{\star}$ uses.
  By constraints~\eqref{eq:pack}, all jobs get assigned, as they are assigned before the additional jobs corresponding to virtual machines.
  We always have at least as many virtual machines as $\bm{x}^\star$, because of the extra configuration added to a fastest machine.
  We thus need to guarantee that the additional speed scheduled on a fastest machine is sufficiently bounded.
  The variables of~$\bm{x}^{\star}$ corresponding to the fastest $L$ machines are already integral. 
  Hence, the first $b_i$, for which additional speed is put onto the fastest machine, is at most $\delta^3\cdot s_{\max}/\log(1/\delta)$.
  There are $\lceil\log(1/\delta)\rceil\lambda+1$ constraints on the variables~$x_{i,\bm{\gamma}}$ for $\bm{\gamma}\in\mathcal{C}_i$.
  Consequently, at most that many variables $x_{i,\gamma}$ can be positive in a vertex solution of the projected LP, which we can find in polynomial time by \autoref{lem:MILP:basic}.
  As each of these has a size $b_i$, the total additional speed assigned to a fastest machine is bounded by
  \begin{align*}
  &s_{\max}(\lceil\log(1/\delta)\rceil\lambda+1)\sum_{r=L}^{\infty}b_{r}=s_{\max}(\lceil\log(1/\delta)\rceil\lambda+1)\sum_{k=L/\lambda}^{\infty}\sum_{\ell=0}^{\lambda-1}(2-\ell/\lambda)2^{-k}\\
  &\qquad\qquad\leq\frac{\delta^3\cdot s_{\max}}{\log(1/\delta)}(\lceil\log(1/\delta)\rceil\lambda+1)(1+3\lambda)<5\delta s_{\max}\enspace.
  \end{align*}
  The last inequality holds for $\delta\leq1/10$.
  Adding this much speed to a fastest machine results in a schedule with makespan at most $(1+5\delta)\cdot T$.
  The run time needed to construct the schedule is bounded by the number of machines times the effort per machine, resulting in
  \begin{align*}
    \mathcal{O}(N\cdot\tau)=\mathcal{O}(1/\delta \cdot N\log(1/\delta \cdot N^2))\subseteq 2^{\mathcal{O}(1/\delta)}+\mathcal{O}(N\log^2(N))\enspace.\tag*{\qedhere}
  \end{align*}    
\end{proof}

\begin{toappendix}

\section{Faster Schedule Construction}
\label{sec:fasterSched}
The results of the previous section already give us an EPTAS for $Q||C_{\max}$ with almost linear run time of $2^{\mathcal{O}(1/\eps \log^{3}(1/\eps) \log(\log(1/\eps)))}+\mathcal{O}(N\log^2(N))$.
Interestingly, the bottleneck (with respect to $N$) of this approach is the transformation from a valid MILP solution into a feasible schedule.
In this section, we give a more conventional $Q||C_{\max}$ MILP formulation \eqref{hybrid-MILP} using both configuration and assignment variables to improve the run time in $N$.
First, we give an algorithm to transform a solution of \eqref{recursive-MILP} into a solution of \eqref{hybrid-MILP} in sublinear run time in $N$.
Then, we show how to construct a schedule from a solution to \eqref{hybrid-MILP} in linear time.
This allows us to transform a solution of \eqref{recursive-MILP} into a valid schedule in linear run time in $N$.
We also use the \eqref{hybrid-MILP} formulation in the high-multiplicity setting.%

Note that \autoref{MILPround1} constructs a solution to \eqref{recursive-MILP} from a schedule to $Q||C_{\max}$.
Hence, both formulations are equivalent up to a multiplicative error of $1+\mathcal O(\eps)$.

\subsection{The Hybrid-MILP Formulation}
Let $\mathcal{C}'_i\coloneqq \{\gamma\in\mathbb{N}^\tau\mid\bm{\gamma}\cdot\bm{b}\leq b_i\text{ and }\supp(\bm{\gamma})\subseteq H_i\}$ be the set of configurations of long jobs with range $1/\delta$ for machine $i$.
For any machine speed $b_i$ and corresponding configuration $\gamma\in\mathcal{C}'_i$, let $\free(i,\gamma)\coloneqq b_i-\bm{\gamma}\cdot\bm{b}$ be the speed of the machine that is free after placing the jobs specified by $\gamma$.
Then \eqref{hybrid-MILP} is given by
\begin{align}
  \sum_{\bm{\gamma}\in \mathcal{C}'_{i}} x_{i,\bm{\gamma}}&= \mu_i&\for i = 1,\ldots,\tau &&\label{hybrid-MILP:1}\\
  \sum_{i=1}^{\tau}\sum_{\bm{\gamma}\in \mathcal{C}'_{i}}\gamma_j\cdot x_{i,\bm{\gamma}}+\sum_{i=1}^{\tau} y_{i,j}&=\eta_j&\for j = 1,\ldots,\tau &&\label{hybrid-MILP:2}\\
  \sum_{\bm{\gamma}\in \mathcal{C}'_{i}}\free(i,\bm{\gamma})\cdot x_{i,\bm{\gamma}}-\sum_{j=1}^{\tau}b_j\cdot y_{i,j}&\geq0 &\for i = 1,\ldots,\tau&&\label{hybrid-MILP:3}\\
  x_{i,\bm{\gamma}},\;y_{i,j}&\geq 0\;&\for\; i =1,\ldots,\tau\;;j = 1,\ldots,\tau\;;\gamma\in\mathcal{C}_i\nonumber\\
  y_{i,j}&=0\;&\for\; i = 1,\ldots,\tau\;;j \in \{\min(H_i),\ldots,\tau\}\nonumber\\
  x_{i,\bm{\gamma}}&\in\mathbb{Z}_{\geq0}\;&\for\; i = 1,\ldots,L\;;\gamma\in\mathcal{C}_i\enspace.&&\label{hybrid-MILP}\tag{hybrid-MILP}
\end{align}
In this formulation, there are no recursive configurations.
Instead, we use configuration variables~$x_{i,\gamma}$, indicating how often a configuration $\gamma$ is used on machine~$i$.
Short jobs, taking up speed less than $\delta b_i$ on machine $i$, are handled via assignment variables~$y_{i,j}$ indicating how many jobs of size $b_j$ are assigned to machines of speed $b_i$.
The constraints \eqref{hybrid-MILP:1} enforce that every machine is assigned a configuration, the constraints \eqref{hybrid-MILP:2} guarantee that every job is scheduled somewhere, and the constraints \eqref{hybrid-MILP:3} make sure that the speed used by short jobs is at most the speed left free by configurations.

We now show how to convert a solution $\bm{x}^{\star}$ of \eqref{recursive-MILP} into a solution $(\bm{x},\bm{y})$ of \eqref{hybrid-MILP}.
\begin{printalgo}
\pcb[head=\parbox{\textwidth-8.8pt}{\textsf{ConvertMILPsolution\\
             Input: A feasible solution $\bm{x}^\star$ to \eqref{recursive-MILP}.\\
             Output: A feasible solution $\bm{x},\bm{y}$ to \eqref{hybrid-MILP}}}]{\textsf{ConvertMILPsolution}}{\label{algoconvert}.
    \t \pcfor \text{$i=1,\ldots,\tau$: initialize $\eta'_i\coloneqq \eta_i$ and $x_{i,\gamma}=x^\star_{i,\gamma}$ for $\gamma\in\mathcal{C}_i$, otherwise $x_{i,\gamma}=0$}\\
    \t \pcfor \text{decreasing machine speeds $b_{i}$}:\\
    \t[2] \text{decrease arbitrary $x_{i,\gamma}>0$ until $\sum_{\gamma\in\mathcal{C}'_i}x_{i,\gamma}=\mu_i$}\\
    \t[2] \pcdo \text{ count the number of jobs $\zeta_j$ in configurations $x_{i,\gamma}$ for $j\in H_i$}\\
    \t[3] \pcfor \text{every job size $b_j$ with $\zeta_j\geq\eta'_j$:}\\
    \t[4] \text{substitute $\zeta_j-\eta'_j$ many jobs in appropriate $x_{i,\gamma}$ with $x_{j,\gamma'}$}\\
    \t[4] \text{(by decreasing $x_{j,\gamma'},x_{i,\gamma}$ and increasing $x_{i,\gamma''}$,}\\
    \t[4] \text{where $\gamma''$ is $\gamma$ with jobs $j$ replaced with $\gamma'$}\\
    \t[4] \text{and jobs shorter than $\delta b_i$ in $\gamma'$ dropped)}\\
    \t[2] \pcuntil \text{there are no more $\zeta_j\geq\eta'_j$}\\
    \t[2] \text{decrease every $\eta'_i$ by $\zeta_j$}\\
    \t \pcfor \text{increasing machine speeds $b_{i}$}:\\
    \t[2] \text{calculate the free machine speed $z_i\coloneqq \sum_{\gamma\in\mathcal{C}'_i}\free(i,\gamma)\cdot x_{i,\gamma}$}\\
    \t[2] \text{starting with the shortest jobs $b_j$, increase $y_{i,j}$, decrease $z_i$ and $\eta'_j$}\\
    \t[2] \pcuntil \text{$z_i$ is $0$, $\eta'_j=0$ or $j\in H_i$}\\
    \t \pcreturn \text{$(\bm{x},\bm{y})$}
}
\end{printalgo}
\begin{lemma}\label{lem:MILP:conversion}
  \autoref{algoconvert} converts a solution $\bm{x}^\star$ of \eqref{recursive-MILP} into a solution $(\bm{x},\bm{y})$ of \eqref{hybrid-MILP} in time $2^{\mathcal{O}(1/\delta\log^2(1/\delta))}\log^{\mathcal{O}(1)}(N)$.
\end{lemma}
\begin{proof}
  The algorithm guarantees that the sum of the configuration variables $x_{i,\gamma}$ for machines with speed~$b_i$ is exactly $\mu_{i}$.
  Hence, constraints \eqref{hybrid-MILP:1} are satisfied.
  Jobs are only replaced once all original jobs have been accommodated and, by constraints \eqref{eq:virtual}, there are exactly as many additional virtual machines, as there are additional jobs.
  Therefore, this step never reduces the total number of configurations for any machine speed below $\mu_i$.
  After the loop of lines 2-11 on machines with speed $b_i$, only jobs of size $b_j \in (\delta b_i,b_i]$, or in other words $j\in H_i$, are assigned via configurations. %
  Additional short jobs would have been assigned via the recursive configurations on these machines, which we neglect by stopping at $\delta b_i$.
  Thus, these machines have sufficient speed to handle these short jobs, which would have been assigned to them.
  Instead of their original order, we assign the short jobs by increasing size.
  This does not change the total load assigned, which therefore still remains sufficient.
  As the sorting ensures that any short job assigned is smaller or equal to some short job that would have been assigned by the original order, we do not assign jobs  which are no longer tiny.
  Due to having sufficient machine speed and not dropping jobs anywhere in the algorithm, we assign all real jobs and hence satisfy constraints \eqref{hybrid-MILP:2}.
  Finally, constraints~\eqref{hybrid-MILP:3} are satisfied by construction, as we never overfill the available speed.
  The number of configurations~$|\mathcal{C}'_i|$ is bounded by $2^{\mathcal{O}(1/\delta\log^2(1/\delta))}$ and all other quantities are bounded by $\mathcal{O}(1/\delta \log(1/\delta \cdot N))$.
  That allows us to analyze the run time of the algorithm to be within the claimed complexity.
\end{proof}

\subsection{Constructing a Schedule}
Now, we have an assignment for all small jobs, albeit with fractional variables. %
In this format we can assign the small jobs integrally more quickly then if we had to resolve recursive configurations.
\begin{lemma}
\label{lem:hybMILP:schedule}
  Given a feasible solution of \eqref{hybrid-MILP}, a schedule with makespan at most $(1+9\delta)T$ can be constructed in time $2^{\mathcal{O}(1/\delta\log^2(1/\delta))}+\mathcal{O}(N)$.
\end{lemma}
\begin{proof}
  For the configuration variables, we pursue the same strategy as in \autoref{lem:recMILP:schedule}, that is, rounding them down and assigning one configuration to a fastest machine for every rounded variable.
  Through the use of basic solutions, we construct a schedule introducing a multiplicative error of at most$(1+5\delta)$ in comparison to the optimal makespan.
  The process takes time $2^{\mathcal{O}(1/\delta\log^2(1/\delta))}\log^{\mathcal{O}(1)}(N)$ with our increased number of configurations.
  
  For the assignment variables of the short jobs, we first note that any machine speed $b_i$ is assigned at most $2$ fractional short jobs, as a variable only becomes fractional when the preceding or current group runs out of speed.
  As in \autoref{lem:MILP:conversion}, we first sort the assignment variables, in time $\mathcal{O}((1/\delta \log(1/\delta \cdot N))^4)$.
  By increasing every machine speed by a factor of $(1+2\delta)$, those two fractional jobs can be placed on an arbitrary machine of speed $b_i$, without exceeding the speed.%
  We get another overhead of a factor of $(1+\delta)$ by greedily packing the assigned jobs to machines, overpacking each machine just slightly.
  This greedy packing takes time $\mathcal{O}(N+(1/\delta\log(1/\delta \cdot N))^2)$.
  In total the approximation error is bounded by $(1+5\delta)(1+2\delta)(1+\delta)\leq (1+9\delta)$ for $\delta<1/100$, as claimed.
\end{proof}

\subsection{The Complete Algorithm}
Finally, we have all the components of our EPTAS, and can therefore show our original claim of having only linear dependence on $N$.
\NewQCmaxEPTAS*
\begin{proof}
  In total, we build \eqref{recursive-MILP}, solve it, convert the solution into one of \eqref{hybrid-MILP} and finally construct a schedule.
  Throughout \autoref{sec:pre}, we proved that the preprocessing takes time bounded by $\mathcal{O}(1/\delta^6)+\mathcal{O}(N)$.
  Finding a solution to \eqref{recursive-MILP} in \autoref{sec:QCmaxMILP} contributes the dominating run time term $2^{\mathcal{O}(1/\delta\log^3(1/\delta)\log(\log(1/\delta)))}$ in terms of $1/\delta$.
  With the techniques of \autoref{sec:RunTime}, all $\log^{\mathcal{O}(1)}(N)$ run time factors are absorbed.
  Using \autoref{MILPround1}, we find a feasible solution with approximate makespan to \eqref{recursive-MILP}.
  
  The following steps are only performed for the feasible \eqref{recursive-MILP} solution $\bm{x^\star}$ with the smallest makespan, and therefore do not incur the $\mathcal{O}(1/\delta \log(N))$ run time factor of the binary search.
  We convert $\bm{x^\star}$ into a feasible solution of \eqref{hybrid-MILP}, in time $2^{\mathcal{O}(1/\delta\log^2(1/\delta)\log(\log(1/\delta)))}+\mathcal{O}(N)$, through \autoref{lem:MILP:conversion}.
  Finally, \autoref{lem:hybMILP:schedule} gives a schedule to the original instance, with approximately optimal makespan in linear run time $\mathcal{O}(N)$.
  After substituting $\delta$ with $\eps/c$, for the appropriate algorithm-dependent constant $c$, the claimed result follows.
\end{proof}

\section{High-Multiplicity Scheduling}
So far, we only considered scheduling problems, where each job and each machine are given as a separate datum in the input, even though many of them might be identical.
More concretely, we have not distinguished between the total number of jobs $N$ and the number of \emph{distinct} jobs $n$, as well as the total number of machines $M$ and the number of \emph{distinct} machines $m$.
Clearly, representing the input with regard to $N$ and $M$ gives a very inefficient input encoding.
For example, consider an instance of $P||C_{\max}$, where \emph{all} $N$ jobs have processing time $p$ for some fixed $p \in \mathbb N$.
The previous encoding would represent the jobs of this instance as a list of $N$ times the value $p$, and the encoding size would be $\mathcal{O}(N\cdot \log(p))$.
In contrast, we can also use a more efficient encoding, called \emph{high-multiplicity scheduling}, which represents this instance as the tuple $(p,N)$ leading to a much smaller encoding size of $\mathcal{O}(\log(N)+\log(p))$.
Such efficient encodings were first studied for scheduling problems on a single machine~\cite{BraunerCGK2005,GranotS1993,Hamada2010,HochbaumS1991,HochbaumSS1992,HochbaumS1990,Psaraftis1980}; later on, these encodings were also considered for scheduling on multiple machines~\cite{GoemansR2020,McCormickSSF2001}.
We note here, that efficient algorithms for high-multiplicity encodings with multiple machines need to use carefully constructed outputs, as a naive assignment of jobs to machines would need space (and thus time) $\Omega(M)$, which might be exponential. 
Exact algorithms for high-multiplicity scheduling have recently attracted interest from a parameterized complexity point of view~\cite{BrinkopJ2022,KnopKLMO2019, KouteckyZ2020}. 

The previous, inefficient, encoding gave us two lists of $N$ and $M$ elements.
In designing efficient algorithms, we thus could treat both $N$ and $M$ as if given in unary.
In the high-multiplicity scenario, we are given two lists of $n$ and $m$ elements.
Each list element contains two values: the processing time $p_{j}$ (resp.~the speed $s_{i}$) and the number of jobs (resp.~machines) having this processing time (resp.~speed).
In this setting, $n$ and $m$ still act like they are unary, but $N$ and $M$ are effectively given in binary.
Hence, we need to consider bounds with regard to $n,m$ and $\log(N),\log(M)$ for efficient algorithms.

This more efficient encoding allows for two natural generalizations of our original problem.
\begin{definition}
  The problem \emph{$Q|\mathit{HM}_J|C_{\max}$} is the problem $Q||C_{\max}$ with the jobs given in the format $\mathcal{J}\coloneqq \{(p_j,\eta_j)\in\mathbb{N}^2\mid j = 1,\ldots,n\}$, where $\eta_j$ copies of job $j$ have processing time~$p_j$.
\end{definition}
\begin{definition}
  The problem \emph{$Q|\mathit{HM}|C_{\max}$} is the problem $Q|\mathit{HM}_J|C_{\max}$, where  machines are also given in the format $\mathcal{M}\coloneqq \{(s_i,\mu_i)\in\mathbb{N}^2\mid i = 1,\ldots,m\}$, where  $\mu_i$ copies of machine $i$ have speed $s_i$.
\end{definition}
Hence, an efficient algorithm for instances of $Q|\mathit{HM}|C_{\max}$ (which we call the \emph{general high-multiplicity setting}) needs to run in time polynomial in $n$ and $m$.
The third combination, $Q|\mathit{HM}_M|C_{\max}$, where only machines are given in high-multiplicity, immediately reduces to the normal representation, because of $M\leq N$, and is therefore irrelevant.

As usual for very compact input encodings, we need to pay special attention to the representation of a solution.
For example, if we use the intuitive representation of a schedule $\sigma\colon J\to M$ via its function table, our solution has a size of at least $N\log(M)$, which is \emph{exponential} in the input size of a high-multiplicity encoding.
Hence, we cannot require an assignment for every individual job.
In the $\mathit{HM}_J$ setting, it is natural to also consider a high-multiplicity output instead, i.e., one where we can represent a schedule by each machines number of jobs of each type.
This is equivalent to returning a configuration for each machine, which then represent a schedule for all jobs together.
Such an output has a size of at least $Mn\log_2(N)$, which is still polynomial in this setting.
Filippi and Romanin-Jacur~\cite{FilippiRJ2009} presented an approximation algorithm with additive error of $p_{\max}$ for $P|\mathit{HM}_J|C_{\max}$, using this output format.

We first show how to represent a solution efficiently in the general high-multiplicity setting by allowing the configurations to have multiplicities as well.
Using this solution representation, the output has size at most $m\log(M)n\log(N)$, which is polynomial in the input.
Using this encoding, we show that the algorithm by Filippi and Romanin-Jacur~\cite{FilippiRJ2009} can be generalized to the $P|\mathit{HM}|C_{\max}$ setting and yields an additive error of at most $p_{\max}$.
Clearly, this is a $2$-approximation for $P|\mathit{HM}|C_{\max}$.
We show that it can also be turned into a $(2+\eps)$-approximation algorithm for $Q|\mathit{HM}|C_{\max}$.
To our knowledge, this is the first result of this type in the general high-multiplicity setting.
Finally, we show that our techniques also allow for an EPTAS for $Q|\mathit{HM}|C_{\max}$ with the current state of the art run time, in terms of $\eps$.
This suggests that a high-multiplicity encoding is no significant obstacle to efficient approximation algorithms for $Q||C_{\max}$.

\subsection{Approximation Algorithms for High-Multiplicity Scheduling}
Filippi and Romanin-Jacur~\cite{FilippiRJ2009} show how to find a schedule on uniformly related machines with an additive error of $p_{\max}/s_{m-1}$ on the optimal makespan:
\begin{proposition}[{Filippi, Romanin-Jacur~\cite[Thm. 9]{FilippiRJ2009}}] There is an algorithm for $Q|\mathit{HM}_J|C_{\max}$ which gives a solution with makespan $\OPT+p_{\max}/s_{m-1}$, where $s_{m-1}$ is the second-slowest machine speed, in time $\mathcal{O}(n\cdot M)$.
\end{proposition}
We generalize their result by \autoref{hm:algo}, with which we can now show \autoref{thm:hmQp}.
It uses the optimal load of the fractional relaxation as a lower bound and assigns longest jobs to fastest machines, overpacking each machine by at most one job.
The crux is that the assignment is performed in a block-wise manner, handling as many jobs or machines as possible at once, by utilizing the high-multiplicity representation of both the input and the output.
\begin{printalgo}
\pcb[%
         head={\parbox{\textwidth-8.8pt}{$\mathit{HM}$-\textsf{GreedyApproximation}\\
             Input: An instance $\mathcal{I}\coloneqq (\{(p_1,\eta_1),\ldots,(p_n,\eta_n)\},\{(s_1,\mu_1),\ldots,(s_m,\mu_m)\})$ and\\ a preemptively feasible makespan $T$.\\
             Output: A set of assigned configurations $\sigma\coloneqq \{\ldots,((s_i,a_{i,j}),\{\ldots,(p_j,b_{i,j}),\ldots\}),\ldots\}$.}}]{$\mathit{HM}$-\textsf{GreedyApprox}}{\label{hm:algo}
        \t \text{set $i\coloneqq 1,j\coloneqq 1,R\coloneqq \{\}$} \pccomment{machine- and job-index, result set of configurations}\\
        \t \pcwhile \text{$j\leq n$}: \pccomment{as long as there are jobs left to schedule}:\\
        \t[2] \text{set $x\coloneqq \lceil s_i T/p_j\rceil$} \pccomment{number of jobs fitting on current machine}\\
        \t[2] \pcif \text{$\mu_i x<\eta_j$}: \pccomment{Are we machine limited?}\\
        \t[3] \text{add $((s_i,\mu_i),\{(p_j,x)\})$ to $R$} \pccomment{Yes, then fill up all machines}\\
        \t[3] \text{set $\eta_j\coloneqq \eta_j-\mu_ix$ and $\mu_i\coloneqq 0$} \pccomment{and remove the assigned jobs and machines}\\
        \t[2] \pcelseif \text{$x<\eta_j$}: \pccomment{Enough jobs for multiple machines?}\\
        \t[3] \text{add $((s_i,\lfloor \eta_j/x\rfloor),\{(p_j,x)\})$ to $R$} \pccomment{Yes, then fill up machines}\\
        \t[3] \text{set $\eta_j\coloneqq \eta_j-\lfloor \eta_j/x\rfloor x$ and $\mu_i\coloneqq \mu_i-\lfloor \eta_j/x\rfloor$} \pccomment{and remove assigned jobs and machines}\\
        \t[2] \pcelse \pccomment{Otherwise we have to fill a single machine}\\
        \t[3] \text{set $\gamma\coloneqq ((s_i,1),\{\})$ and $C\coloneqq 0$} \pccomment{Configuration and load for machine}\\
        \t[3] \pcwhile \text{$C<s_iT$ and $j\leq n$}: \pccomment{Repeat until full or out of jobs}\\
        \t[4] \text{set $y\coloneqq \min(\lceil(s_iT-C)/p_j\rceil,\eta_j)$} \pccomment{Max addable amount of $p_j$}\\
        \t[4] \text{add $(p_j,y)$ to $\gamma$}\\
        \t[4] \text{set $\eta_j\coloneqq \eta_j-y$ and $C\coloneqq C+p_jy$} \pccomment{remove and count assigned}\\
        \t[4] \pcif \text{$\eta_j=0$: set $j\coloneqq j+1$} \pccomment{Go to next job class when current is used up}\\
        \t[3] \pcendwhile\\
        \t[3] \text{add $\gamma$ to $R$} \pccomment{Add our single machine to the result}\\
        \t[3] \text{set $\mu_i\coloneqq \mu_i-1$} \pccomment{and remove it from the available ones}\\
        \t[2] \pcendif\\
        \t[2] \text{$\textbf{if}\;\mu_i=0$: set $i\coloneqq i+1$} \pccomment{Go to next machine class when current is used up}\\
        \t[2] \text{$\textbf{if}\;\eta_j=0$: set $j\coloneqq j+1$} \pccomment{Go to next job class when current is used up}\\
        \t \pcendwhile\\
        \t \pcreturn \text{R} \pccomment{Returns the set of configurations}
}
\end{printalgo}
\GreedyHMQCmax*
\begin{proof}
  The crucial observation in the analysis of \autoref{hm:algo} is that every creation of a block of single job class configurations either finishes a job or machine class outright, or must be followed by the creation of a single machine configuration.
  There are two cases to consider:
  First, the remaining volume of the current job class $j$ fills up all space on the machines of the current machine class~$i$.
  In this case, we use up this machine class.
  Second, the remaining volume of the current job class $j$ does not fill up all space on the machines of the current machine class $i$.
  In this case, we used up this job class and continue to fill the current machine with the next job class. 
  Therefore, the outer while loops runs at most twice for every machine class and job class, or more precisely at most $2n+2m-1$ many times.
  The inner while loop can have more than one iteration, but in doings so either fills the current machine or uses up a complete job class. 
  Therefore, the inner while loop can at most steal iterations from the outer loop, but not increase the total number.
  Hence, the number of configurations produced is also bounded by $\mathcal{O}(n+m)$, as each iteration of the outer while loop produces at most one configuration.
  
  The algorithm never runs out of machines before it has assigned all jobs, as it assigns to each machine at least the average load, before moving to the next one.
  On the other hand, it overloads each machine with at most one job, and clearly the average completion time is a lower bound on the optimal makespan.
  Hence, the algorithm gives an approximate schedule with an additive error of at most $p_{\max}$.
\end{proof}
In the $P|\mathit{HM}|C_{\max}$ case (where all machines have the same speed), this schedule can easily be seen to be $2$-approximate, as $p_{\max} \leq \OPT$ holds by setting $T\coloneqq \sum_{j=1}^n p_j\eta_j/\sum_{i=1}^m s_i\mu_i$.
Unfortunately, for $Q|\mathit{HM}|C_{\max}$, using $T$ as a lower bound can have an arbitrarily bad approximation ratio, as we can only guarantee $p_{\max}/s_{\max} \leq \OPT$.
However, this bad approximation ratio can only happen due to the scheduling of a single big job on a machine. 
Otherwise, the overhanging jobs would have to be smaller than the makespan to even be assigned.
We will thus make use of an observation by Woeginger~\cite{Woeginger1999}, that we have $\OPT\leq 2\OPT_{\text{pmtn}}$, where $\OPT_{\text{pmtn}}$ is the optimal makespan of the preemptive relaxation (i.e., where we allow interruption and resumption of jobs on different machines).
This observation is equivalent to the fact that $T_{p}\leq 2\OPT$, where $T_p=\max\left(\sum_{j=1}^N p_j/\sum_{i=1}^M s_i,\max_{k=1,\ldots,M}\sum_{j=1}^k p_j/\sum_{i=1}^k s_i\right)$.
Note that the quantity $T_p$ is the maximum of the average load and the average loads of the $k$ longest jobs on the~$k$ fastest machines.
Computing~$T_{p}$ naively leads to the evaluation of $M+1$ many sums.
However, we can also compute $T_{p}$ efficiently in the high-multiplicity setting by realizing that $(a+x\cdot p)/(b+x\cdot s)$ is monotone in~$x$ (as its derivative is $(bp-as)/(b+sx)^2$).
Thus, the maximum $\max_{k=1,\ldots,N}\sum_{j=1}^k p_j/\sum_{i=1}^k s_i$ can only occur after having added none, or all of the available jobs or machines of a size or speed respectively.
Hence, we only need to consider values of $k$, where all jobs of a certain size were added or all machines of a certain speed.
As there are only $n$ different job sizes and $m$ different machine speeds, we only need to consider $\mathcal{O}(n+m)$ different values for $k$.
The corresponding partial sums can then be calculated in time linear in the high-multiplicity encoding.

Now, using $T_{p}$ directly as our makespan guess does not work yet, as our algorithm can still place a big job badly, creating a makespan of $T_{p}+p_{\max}$.
Using a makespan guess of $2T_p$ however is sufficient to compute a $4$-approximate schedule: A single large job put onto the fastest remaining machine and overloading this machine would lead to a contradiction, as we would now have allocated more longer jobs on faster machines then the optimal makespan.
This implies that our algorithm constructs a $2$-approximate schedule, if we set $T$ to the (unknown) optimal makespan $\OPT$.
By performing a binary search with a precision of $\eps$ in the interval $[T_p,2T_p]$, we only make  $\mathcal{O}(\log_2(\log_{1+\eps}(2)))=\mathcal{O}(\log(1/\eps))$ calls to \textsf{$Q|\mathit{HM}|C_{\max}$GreedyApproximation} and will obtain a guaranteed $(2+\eps)$-approximate schedule.
Hence, we have proved \autoref{cor:hmQ2}.

As summarized in the proof of \autoref{thm:newqcmaxeptas}, all preconditioning steps of our algorithm are already efficient for high-multiplicity encodings.
Afterwards, the construction of the quantities $\bm{\eta},\bm{\mu}$ directly gives a high-multiplicity encoding of the rounded instance.
Solving \eqref{recursive-MILP} thus also had a run time depending on $\log(N)$ and not on $N$.
The only problem arises in the construction of the schedule, as
the classical greedy algorithm has time linear in the number of jobs.
By replacing this part with \autoref{hm:algo}, we achieve the same additive $p_{\max}$ approximation.
We can therefore generate an approximate schedule of configurations with multiplicities through this adaption, proving \autoref{cor:hmQe}.

\section{The Transition From Few to Many Machine Speeds}
\label{sec:transition}
In the construction of the fastest EPTAS for $P||C_{\max}$, Berndt et. al.~\cite{BerndtDJR2022} applied an ILP solver by Jansen and Rohwedder~\cite{JansenR2018,JansenR2022}, optimized for programs with few constraints, to obtain a run time of  $2^{\mathcal{O}(1/\eps\log(1/\eps)\log(\log(1/\eps)))}+\mathcal{O}(N)$.
For $Q||C_{\max}$, we resorted to using an MILP algorithm, instead of the optimized ILP algorithm.
Ultimately, this gave us a factor of $\log^3(1/\eps)$, instead of the $\log(1/\eps)$ for identical machines, in the exponential run time term.
Specifically, the ILP algorithm is:

\begin{proposition}[Jansen, Rohwedder~\cite{JansenR2018,JansenR2022}]
\label{ILP:Rohwedder}
  For any instance of $ILP$, we can either find a solution $\bm{x}\in\mathcal{Q}$ or determine that $\mathcal{Q}=\emptyset$ in time $2^{\mathcal{O}(m\log(A_{\max}))}\langle \mathit{ILP}\rangle^{\mathcal{O}(1)}$.
\end{proposition}

In the following, we present an algorithm whose run time transitions continuously between the state-of-the-art run times for approximation schemes for $P||C_{\max}$ (all machines have identical speeds) and $Q||C_{\max}$ (there are no restrictions at all on the number of distinct machine speeds), depending on the structure of the given machine speeds.
More concretely, we can improve our general run time if there is a block $B$ of relatively few machines with speeds in a range~$1/\delta$ among the fast, but not too fast machines.
This allows us to separate the problem into two separated instances, split along this block~$B$.
Machines faster than those in $B$ are handled by an ILP formulation of \eqref{recursive-MILP}, which is solved with \autoref{ILP:Rohwedder}.
Machines slower than those in $B$ are handled via an LP algorithm, and the solution of the ILP and the LP are combined into an approximately optimal schedule of our original input.

Recall that after rounding, we index the machine speeds by $\{1,\ldots,\tau\}$ and the $L\in \mathcal{O}(1/\delta \log(1/\delta))$ fastest machines, indexed by $\{1,\ldots,L\}$, are those that are packed integrally in \eqref{hybrid-MILP}.
Additionally, $H_{i}=\{j = 1,\ldots,\tau \mid b_{j}\in (\delta b_{i},b_{i}]\}$ are the set of indices of long job speeds, i.e., those that use at least a fraction of $\delta$ of the complete machine speed.
Note that all the corresponding variables $x_{i,\gamma}$ for $i\in H_i$ are fractional in \eqref{hybrid-MILP}, if $i > L$.
The set of corresponding speeds is given by $B_{i} \coloneqq \{b_j\mid j\in H_i\}$, which we call a \emph{block}.
Simultaneously, $B_{i}$ is also the set of job processing times which can be assigned to machines with speed $b_i$ without exceeding the makespan guess.
Such a block contains $m_i \coloneqq \sum_{j\in H_i}\mu_j$ machines in total.

If there is an $i\in \{L+1,\ldots,\tau\}$ such that the number of machines $m_{i}$ is small, we can split up the instance into two separate parts.
As first step, we guess the vector $\bm{h}_{i}$ describing the assignment of long jobs to machines with speed $b_{i}$.
If $m_{i}$ is sufficiently small, this can be done reasonably fast.
We thus assume that we have guessed $\bm{h}_{i}$ correctly.
Now, we can modify \eqref{hybrid-MILP} by hardcoding the assignment given by $\bm{h}_{i}$ into it.
Note that this also implies hardcoding a specific amount of space on machines in group $i$ available for tiny jobs.
As a consequence, for solutions of this modified \eqref{hybrid-MILP}, the assignments to machines with speeds faster than $B_{i}$ and the assignments to machines with speeds slower than $B_{i}$ barely interact anymore.
In fact, their only interaction is given by the free space needed on the faster machines that will be used for tiny jobs.
Since this free space might be large, we can not simply guess it.
But, due to our rounding, this free space is a multiple of $\rho_i^{\text{\tiny}}$.
Hence, we know that a space of $N \cdot \rho_{i}^\text{tiny}$ suffices to place all of these tiny jobs.
The appropriate multiple $\free_{i}$ of $\rho_{i}^\text{tiny}$ can thus be found via binary search by ensuring that the free space is at least as large as the usable space.
Implementing this approach requires a small modification to \eqref{hybrid-MILP}.
Originally, we checked globally that the tiny jobs only use the existing free space via the inequality
\begin{align*}
  \sum_{\bm{\gamma}\in \mathcal{C}'_{i}}\free(i,\bm{\gamma})\cdot x_{i,\bm{\gamma}}-\sum_{j=1}^{\tau}b_j\cdot y_{i,j}&\geq0 &\for i = 1,\ldots,\tau\enspace.
\end{align*}
Instead, we now do this cumulatively and check the space in each group in  $\{1,\ldots,k\}$ for $k = 1,\ldots,\tau$.

To do this, we make extensive use of the following observation:
A job type is long in up to two groups, and these corresponding groups need to be consecutive.

Now, given $\bm{h}_{i}$ and $\free_{i}$, our problem decomposes into two independent problems. 
All variables corresponding to faster machines with speeds $b_{1},\ldots,b_{i}$ will be solved integrally by applying \autoref{ILP:Rohwedder}.
For the remaining machines, all variables will be solved fractionally by an appropriate LP solver.
A combination of these solutions corresponds to a feasible solution of \eqref{hybrid-MILP} (after converting it into a cumulative version) where we simply require integrality for more variables, as $i \geq L$.
Hence, we can compute a corresponding schedule as described before.

We still need to analyze the run time of our approach.
To compute $\bm{h}_{i}$, we have two options.
We can guess each entry of the vector, which is a number between $0$ and~$m_{i}/\eps$, which results in a run time of
\begin{equation*}
  \left(\frac{m_{i}}{\eps} + 1\right)^{\mathcal{O}(1/\eps \log(1/\eps))} \subseteq 2^{\mathcal O(1/\eps \log(1/\eps) (\log (1/\eps) + \log(m_i))} \enspace .
\end{equation*}
Alternatively, we guess the type of each job, which results in a run time of
\begin{equation*}
  \mathcal{O}\left(\frac{1}{\eps} \log\left(\frac{1}{\eps}\right)\right)^{m_{i}/\eps} \subseteq 2^{\mathcal{O} (m_{i}\cdot 1/\eps \log(1/\eps) \log(\log(1/\eps)))} \enspace .
\end{equation*}
Now, we choose $i\geq L$ such that $m_{i}$ is minimal and can compare the run time of both approaches, as all parameters are known at run time.
We can thus guess $\bm{h}_{i}$ in time
\begin{equation*}
  2^{\mathcal{O} \left(1/\eps \log(1/\eps) \min \{ m_i \log(\log(1/\eps)),\,  \log(m_i) + \log(1/\eps) \}\right)} \enspace .
\end{equation*}

As described above, the free space $\free_{i}$ can be guessed in time $\mathcal{O}(\log(N))$ using binary search.
Solving the LP for the slower machines is possible in time $2^{\mathcal{O}(\log^2(1/\eps))} \cdot \log^{\mathcal{O}(1)}(N)$, which is negligible compared to the rest of the algorithm.
Solving the ILP for the faster machines via the algorithm of Jansen and Rohwedder~\cite{JansenR2018,JansenR2022} takes time $2^{\mathcal{O}(i \log(i) \log(\log(1/\eps)))}\cdot \log^{\mathcal{O}(1)}(N)$.
This yields a total run time of
\begin{equation*}
  \min_{i > L}2^{\mathcal{O} \left(1/\eps \log(1/\eps) \min \{ m_i \log(\log(1/\eps)),\,  \log(m_i) + \log(1/\eps) \} + i \log(i) \log(\log(1/\eps))\right)}\cdot \log^{\mathcal{O}(1)}(N) \enspace .
\end{equation*}

In case of identical machines $P||C_{\max}$, we can simply choose $i \coloneqq L+1$, as $m_{i} = 0$ for all $i > 1$.
Doing so, we achieve the state-of-the-art run time~\cite{BerndtDJR2022} of $2^{\mathcal{O}(1/\eps\log(1/\eps)\log(\log(1/\eps)))}+\mathcal{O}(N)$.
This also covers the generalization of a bounded number of distinct machine speeds $k$, as the speeds are either close enough to each other to be solved efficiently or there is some gap and thus a small $m_{i}$.
Balancing these two options results in a run time of $2^{\mathcal{O}(k \cdot 1/\eps \log(1/\eps) \log(\log(1/\eps)))} + \mathcal O(N)$, thereby proving the following theorem.

\SpecialCase*

As a consequence, for $k \leq \log^2(1/\eps)$, either there is some $i \leq k$ with $m_i = 0$ or, by the pigeonhole principle, that for $i \coloneqq k + 1$ we have $m_i = 0$.
Hence, if the number of different machine speeds is at most $\log^2(1/\eps)$, we obtain a run time that, depending on the exact $k$, is somewhere between $2^{\mathcal{O}(1/\eps\log^{3}(1/\eps)\log(\log(1/\eps)))} + \mathcal O(N)$ and $2^{\mathcal{O}(1/\eps\log(1/\eps)\log(\log(1/\eps)))} + \mathcal O(N)$.
This shows that for relatively few different machine types, we can even improve the run time of our algorithm presented in~\autoref{thm:newqcmaxeptas}. 

Furthermore, if there is a machine type that has only very few machines in terms of $N$, we can also improve the run time.
More formally, if there is some $i \leq \log(\log(N))$ where $m_i \leq \mathcal{O}(\log^{\mathcal{O}(1)}(N))$, we can solve this in time with better dependency on $1/\eps$ than the general case of \autoref{thm:newqcmaxeptas}.
If $\log(N) \geq 1/\eps^{2}$, then our generic algorithm has run time $\mathcal{O}(N)$.
Otherwise, we have $2\log (1/\eps) \geq \log(\log N)$ and thus $i \leq 2\log(1/\eps)$.
By taking the smaller term, we obtain a run time of $2^{1/\eps \log^2 (1/\eps) \log (\log (1/\eps))} + \mathcal O(N)$.

\end{toappendix}

\section{Uniform Machine Types}
It is well-known that for the very general $R||C_{\max}$ problem one cannot obtain a $(3/2-\varepsilon)$-approximation in polynomial time for any $\varepsilon > 0$, unless $\mathsf{P}=\mathsf{NP}$~\cite{LenstraST1990}.
In this problem, if job $j$ is processed on machine $i$, it takes time $p_{i,j}$.
As these values do not need to be related (in contrast to the uniform case, where $p_{i,j}=p_{j}/s_{i}$), this problem is often called scheduling on \emph{unrelated machines}.
To understand the approximability of this problem in more depth, Jansen and Maack~\cite{JansenM2019} studied the problem $R_K||C_{\max}$, where there are $K$ \emph{machine types}.
They presented an EPTAS with run time $2^{\mathcal{O}(K\log(K)1/\eps\log^4(1/\eps))}+\langle \mathcal{I}\rangle ^{\mathcal{O}(1)}$.
They also consider the generalization $R_KQ||C_{\max}$, where we have $K$ machine types and for each type, the machines can have different uniform speeds.
For $R_KQ||C_{\max}$, they provided an EPTAS with run time $2^{\mathcal{O}(K\log(K)1/\eps^3\log^5(1/\eps))}+\langle \mathcal{I}\rangle^{\mathcal{O}(1)}$.
To show the generality of our approach, we construct an EPTAS with improved run time of $2^{\mathcal{O}(K\log(K)1/\eps\log^3(1/\eps)\log(\log(1/\eps)))}+\mathcal{O}(K\cdot N)$.
This suggests that adding uniform machines to each machine type does not increase the complexity in approximating the generalized problem.
The exponential improvement in $\eps$ is enabled by our recursive configurations formulation.
It allows us to apply our support size bound globally to the complete MILP.
In contrast, the original approach needs to distinguish simple and complex configurations for each machine type locally.
The improvement of the polynomial dependency in $K\cdot N$ to $\mathcal{O}(K\cdot N)$, which is linear in the input size $\langle\mathcal{I}\rangle$, is the result of analyzing job types instead of individual jobs.

\smallskip\noindent\textbf{Problem Instance Format.}
The input for $R_KQ||C_{\max}$ consists of $N$ jobs and $M$ machines.
Each job has $K$ processing times $p_{j,k}$, given as a matrix $\mathbb{N}^{K\times N}$.
Each machine has one of the $K$ types and a speed~$s_i$, given as $(\{1,\ldots,K\}\times \mathbb{N})^M$.
We generalize the assumption of sorted inputs from $Q||C_{\max}$ and assume that for each machine type, we are also given lists of the jobs and machines, sorted by their processing time and speed respectively.
Such orders can be computed in time $\mathcal{O}(K\cdot N\log(N))$, same as $Q||C_{\max}$ for $K=1$.

\subsection{Preprocessing}
We extended our preprocessing for $Q||C_{\max}$, resulting in job counts $\eta_{j}$, machine counts $\mu_{i,k}$ and a makespan guess $T$.
Again, if $M>N$ for any machine type, we simply discard the slowest machines, hence $M\leq K\cdot N$.
A $N$-approximation to the optimal makespan is found by assigning each job to the machine where it has the shortest completion time.
This approach gives the upper bound $\UB\coloneqq N\cdot\max_{j}\min_{k}p_{j,k}/s_{k,\max}$ (a direct generalization of $N\cdot p_{\max}/s_{\max}$ in $Q||C_{\max}$).
Any processing time $p_{j,k}$ larger than $\UB\cdot s_{k,\max}$ can then be set to $\infty$.
For every machine class, a variant of \hyperref[QCmax:para:Step1]{\textbf{Step 1}} in $Q||C_{\max}$ can be performed, that removes negligible machines and negligible jobs.
More concretely, for each machine class~$k$, jobs are assigned to a fastest machine $s_{k,\max}$, if they are smaller than $\delta\UB s_{k,\max}/N^2$, and machines are dropped if they are slower than $\delta s_{k,\max}/N$.
Hence, we have a ratio of $p_{k,\max}/p_{k,\min}\leq N^2/\delta$.
All of these steps can be performed in time $\mathcal{O}(K\cdot N)$.

By repeating \hyperref[QCmax:para:Step2]{\textbf{Step 2}} for all machine types, we round all inputs in total time $\mathcal{O}(K\cdot(N+1/\delta\log(1/\delta \cdot N)))$.
Afterwards, there are $\tilde{\tau}\coloneqq \log_{1+\delta}(1/\delta \cdot N^2)\in\mathcal{O}(1/\delta \log(1/\delta \cdot N))$ different job speeds left for any machine type.
Hence, there can be at most $\tilde{\tau}^K$ different job types in total.
We count the number of jobs of type $j$ for $j = 1,\ldots,\tau^K$ as $\tilde{\eta}_j$, in time $\mathcal{O}(K\cdot N+\tilde{\tau}^K)$.
For all machine types $k$, the machine speed counts $\tilde{\mu}_{i,k}$ behave as before.
We perform binary search for the makespan as in \hyperref[QCmax:para:Step3]{\textbf{Step 3}} to reduce the optimization into a feasibility problem.
For any machine type, the rounded values $b_r$ can be set up just as before in \hyperref[QCmax:para:Step4]{\textbf{Step 4}}.
From this, we obtain the numbers of machines $\mu_{i,k}$ and jobs $\eta_j$ of each type.
By a similar argument, all preconditioning takes time $\mathcal{O}(1/\delta^{6\cdot K})+\mathcal{O}(K\cdot N)$.

\subsection{Solving an MILP Formulation}
After our preconditioning, we modify the MILP used by Jansen and Maack to accommodate our rounding scheme.
Let $\mathcal{J}_{i,k}\coloneqq \{j \in \{1,\ldots,{\tau^K}\}\mid b_{j_k}=b_i\}$ be the set of indexes of jobs with rounded processing time $b_i$ on machines of type $k$.
Our MILP-formulation for $R_KQ||C_{\max}$ is:
\begin{align}
  \sum_{i'=1}^{\tau}\sum_{\bm{\gamma}\in\mathcal{C}_{i'}}\gamma_i\cdot x_{i',\bm{\gamma},k}-\sum_{j\in \mathcal{J}_{i,k}}y_{j,k}&\geq\sum_{\bm{\gamma}\in\mathcal{C}_i}x_{i,\bm{\gamma},k}-\mu_{i,k}\geq 0&&\for\;i = 1,\ldots,\tau;\;k = 1,\ldots,K\label{uniform-MILP:eq1}\\
  \sum_{k=1}^K y_{j,k}&\geq\eta_j&&\for\;j = 1,\ldots,{\tau^K}\label{uniform-MILP:eq2}\\
  x_{i,\bm{\gamma},k}&\geq0&&\for\;i = 1,\ldots,\tau;\;\bm{\gamma}\in\mathcal{C}_i;\;k = 1,\ldots, K\nonumber\\
  x_{i,\bm{\gamma},k}&\in\mathbb{Z}_{\geq 0}&&\for\;i = 1,\ldots,L;\;\bm{\gamma}\in\mathcal{C}_i;\;k = 1,\ldots, K\nonumber\\
  y_{j,k}&\geq0&&\for\;j = 1,\ldots,\tau;\;k = 1,\ldots, K\enspace.\nonumber\tag{uniform-MILP}\label{uniform-MILP}
\end{align}

We maintain $x$ for the configuration variables, adding an index for the machine class.
The assignment variables $y$ determine the machine class where a job is scheduled.
For the uniform part, described by the constraints \eqref{uniform-MILP:eq1}, we no longer directly give job amounts $\eta_j$, but instead determine the \emph{number of jobs} assigned to a slot size.
Constraints \eqref{uniform-MILP:eq2} ensure that all jobs are assigned to some machine type.

Each variable $y_{j,k}$ only occurs in two constraints, with coefficients $\pm1$, and does therefore not affect the maximal column norm $A_{\max}$.

By \autoref{lem:MILP:support}, if there is a feasible solution for the integral $x_{i,\gamma,k}$ configuration variables, then there is also one with support size $\mathcal{O}(K\cdot1/\delta\log(1/\delta)\log(\log(1/\delta)))$.
There are $K\cdot2^{\mathcal{O}(\log^2(1/\delta))}$ integral configuration variables in total.
Using \autoref{lem:MILP:support-alg}, we find a feasible solution to \eqref{uniform-MILP} in time
\begin{align*}
2^{\mathcal{O}(K\log(K)\cdot 1/\delta\log^3(1/\delta)\log(\log(1/\delta)))}+\mathcal{O}(K\cdot N) \enspace .
\end{align*}

\subsection{Constructing a Schedule}
Now, we have a feasible solution $(\bm{\hat{x}},\bm{\hat{y}})$ to \eqref{uniform-MILP}, from which we need to construct a schedule.
Fixing some solution $\bm{\hat{x}}$ leaves only the assignment variables in the following linear program with constraint matrix $A_{x}$:
\begin{align}
  -\sum_{j\in\mathcal{J}_{i,k}}y_{j,k}\geq c_{i,k}(\bm{\hat{x}})\quad\text{and}\quad\sum_{k=1}^K y_{j,k}\geq \eta_j\qquad y_{j,k}\geq 0\enspace.\label{R-uniform-LP}\tag{R-uniform-LP($\bm{\hat{x}}$)}
\end{align}
As shown by Hoffman and Gale~\cite{HellerT1956}, any matrix, where each column vector has at most the $2$ non-zero opposite sign entries $1$ and $-1$, is totally unimodular.
This is the case for the constraint matrix $A_x$ of \eqref{R-uniform-LP}, as each variable occurs at most once with coefficient $-1$ in constraints \eqref{uniform-MILP:eq1} and once with coefficient $1$ in constraints \eqref{uniform-MILP:eq2}.
Hence, $A_x$ is totally unimodular and any vertex solution is already integral.
For any machine type $k$, we therefore have an integral number of jobs of each size assigned to it and a compatible feasible solution $\bm{\hat{x}_k}$.
Using the techniques of \autoref{sec:fasterSched}, we can construct a schedule $\sigma_k$ for each machine type in time $2^{\mathcal{O}(1/\delta\log^2(1/\delta))}+\mathcal{O}(N)$.
Together, these give a complete schedule $\sigma \coloneqq \bigcup_k\sigma_k$ in time $2^{K\cdot1/\delta\log^2(1/\delta))}+\mathcal{O}(K\cdot N)$, thus proving the following theorem.
\Marten*

\bibliography{References}

\end{document}

%% file: range.tikz
\providecommand\smax{s_{\mathrm{max}}}
\providecommand\pmax{p_{\mathrm{max}}}
\begin{tikzpicture}[scale=2]
	\draw[dashed] (0,0) -- (1,0);
	\draw[dashed] (3,0) -- (4,0);
	\draw[solid] (1,0) -- (3,0);
	\foreach \x in {0,1,3,4} 
	\draw[shift={(\x,0)}] (0pt,3pt) -- (0pt,-3pt);
	\draw[solid] (2,-0.5) -- (4,-0.5);
	\foreach \x in {2,4} 
	\draw[shift={(\x,-0.5)}] (0pt,3pt) -- (0pt,-3pt);

	\draw[decorate, decoration={brace, amplitude=5pt}] (3,0.15) -- (4,0.15);
	\node[anchor=north] at (3.5, 0.5) {$T$};

	\node[anchor=south] at (0,-0.5) {$\smax\delta/N^2$};
	\node[anchor=north] at (1,0.5) {$\pmax \delta/(N\cdot T)$};
	\node[anchor=south] at (2,-1) {$\smax \delta/N$};
	\node[anchor=south] at (4,-1) {$\smax$};
	\node[anchor=north] at (3,0.5) {$\pmax/T$};
\end{tikzpicture}

%% file: rounding.tikz
\begin{tikzpicture}[gnuplot]
\path (0.000,0.000) rectangle (12.500,8.750);
\gpcolor{color=gp lt color axes}
\gpsetlinetype{gp lt axes}
\gpsetdashtype{gp dt axes}
\gpsetlinewidth{0.50}
\draw[gp path] (2.608,0.523)--(11.947,0.523);
\gpcolor{color=gp lt color border}
\gpsetlinetype{gp lt border}
\gpsetdashtype{gp dt solid}
\gpsetlinewidth{1.00}
\draw[gp path] (2.608,0.523)--(2.788,0.523);
\draw[gp path] (11.947,0.523)--(11.767,0.523);
\gpcolor{color=gp lt color axes}
\gpsetlinetype{gp lt axes}
\gpsetdashtype{gp dt axes}
\gpsetlinewidth{0.50}
\draw[gp path] (2.608,0.726)--(11.947,0.726);
\gpcolor{color=gp lt color border}
\gpsetlinetype{gp lt border}
\gpsetdashtype{gp dt solid}
\gpsetlinewidth{1.00}
\draw[gp path] (2.608,0.726)--(2.788,0.726);
\draw[gp path] (11.947,0.726)--(11.767,0.726);
\node[gp node right] at (2.424,0.726) {$\vdots\phantom{b_{2,0}}\;\,$};
\gpcolor{color=gp lt color axes}
\gpsetlinetype{gp lt axes}
\gpsetdashtype{gp dt axes}
\gpsetlinewidth{0.50}
\draw[gp path] (2.608,0.929)--(11.947,0.929);
\gpcolor{color=gp lt color border}
\gpsetlinetype{gp lt border}
\gpsetdashtype{gp dt solid}
\gpsetlinewidth{1.00}
\draw[gp path] (2.608,0.929)--(2.788,0.929);
\draw[gp path] (11.947,0.929)--(11.767,0.929);
\gpcolor{color=gp lt color axes}
\gpsetlinetype{gp lt axes}
\gpsetdashtype{gp dt axes}
\gpsetlinewidth{0.50}
\draw[gp path] (2.608,1.132)--(11.947,1.132);
\gpcolor{color=gp lt color border}
\gpsetlinetype{gp lt border}
\gpsetdashtype{gp dt solid}
\gpsetlinewidth{1.00}
\draw[gp path] (2.608,1.132)--(2.788,1.132);
\draw[gp path] (11.947,1.132)--(11.767,1.132);
\gpcolor{color=gp lt color axes}
\gpsetlinetype{gp lt axes}
\gpsetdashtype{gp dt axes}
\gpsetlinewidth{0.50}
\draw[gp path] (2.608,1.335)--(11.947,1.335);
\gpcolor{color=gp lt color border}
\gpsetlinetype{gp lt border}
\gpsetdashtype{gp dt solid}
\gpsetlinewidth{1.00}
\draw[gp path] (2.608,1.335)--(2.788,1.335);
\draw[gp path] (11.947,1.335)--(11.767,1.335);
\node[gp node right] at (2.424,1.335) {$b_9 = b_{2,0}$};
\gpcolor{color=gp lt color axes}
\gpsetlinetype{gp lt axes}
\gpsetdashtype{gp dt axes}
\gpsetlinewidth{0.50}
\draw[gp path] (2.608,1.843)--(11.947,1.843);
\gpcolor{color=gp lt color border}
\gpsetlinetype{gp lt border}
\gpsetdashtype{gp dt solid}
\gpsetlinewidth{1.00}
\draw[gp path] (2.608,1.843)--(2.788,1.843);
\draw[gp path] (11.947,1.843)--(11.767,1.843);
\node[gp node right] at (2.424,1.843) {$b_8 = b_{1,3}$};
\gpcolor{color=gp lt color axes}
\gpsetlinetype{gp lt axes}
\gpsetdashtype{gp dt axes}
\gpsetlinewidth{0.50}
\draw[gp path] (2.608,2.350)--(11.947,2.350);
\gpcolor{color=gp lt color border}
\gpsetlinetype{gp lt border}
\gpsetdashtype{gp dt solid}
\gpsetlinewidth{1.00}
\draw[gp path] (2.608,2.350)--(2.788,2.350);
\draw[gp path] (11.947,2.350)--(11.767,2.350);
\node[gp node right] at (2.424,2.350) {$b_7 = b_{1,2}$};
\gpcolor{color=gp lt color axes}
\gpsetlinetype{gp lt axes}
\gpsetdashtype{gp dt axes}
\gpsetlinewidth{0.50}
\draw[gp path] (2.608,2.858)--(11.947,2.858);
\gpcolor{color=gp lt color border}
\gpsetlinetype{gp lt border}
\gpsetdashtype{gp dt solid}
\gpsetlinewidth{1.00}
\draw[gp path] (2.608,2.858)--(2.788,2.858);
\draw[gp path] (11.947,2.858)--(11.767,2.858);
\node[gp node right] at (2.424,2.858) {$b_6 = b_{1,1}$};
\gpcolor{color=gp lt color axes}
\gpsetlinetype{gp lt axes}
\gpsetdashtype{gp dt axes}
\gpsetlinewidth{0.50}
\draw[gp path] (2.608,3.365)--(11.947,3.365);
\gpcolor{color=gp lt color border}
\gpsetlinetype{gp lt border}
\gpsetdashtype{gp dt solid}
\gpsetlinewidth{1.00}
\draw[gp path] (2.608,3.365)--(2.788,3.365);
\draw[gp path] (11.947,3.365)--(11.767,3.365);
\node[gp node right] at (2.424,3.365) {$b_5 = b_{1,0}$};
\gpcolor{color=gp lt color axes}
\gpsetlinetype{gp lt axes}
\gpsetdashtype{gp dt axes}
\gpsetlinewidth{0.50}
\draw[gp path] (2.608,4.634)--(11.947,4.634);
\gpcolor{color=gp lt color border}
\gpsetlinetype{gp lt border}
\gpsetdashtype{gp dt solid}
\gpsetlinewidth{1.00}
\draw[gp path] (2.608,4.634)--(2.788,4.634);
\draw[gp path] (11.947,4.634)--(11.767,4.634);
\node[gp node right] at (2.424,4.634) {$b_4 = b_{0,3}$};
\gpcolor{color=gp lt color axes}
\gpsetlinetype{gp lt axes}
\gpsetdashtype{gp dt axes}
\gpsetlinewidth{0.50}
\draw[gp path] (2.608,5.903)--(11.947,5.903);
\gpcolor{color=gp lt color border}
\gpsetlinetype{gp lt border}
\gpsetdashtype{gp dt solid}
\gpsetlinewidth{1.00}
\draw[gp path] (2.608,5.903)--(2.788,5.903);
\draw[gp path] (11.947,5.903)--(11.767,5.903);
\node[gp node right] at (2.424,5.903) {$b_3 = b_{0,2}$ };
\gpcolor{color=gp lt color axes}
\gpsetlinetype{gp lt axes}
\gpsetdashtype{gp dt axes}
\gpsetlinewidth{0.50}
\draw[gp path] (2.608,7.172)--(11.947,7.172);
\gpcolor{color=gp lt color border}
\gpsetlinetype{gp lt border}
\gpsetdashtype{gp dt solid}
\gpsetlinewidth{1.00}
\draw[gp path] (2.608,7.172)--(2.788,7.172);
\draw[gp path] (11.947,7.172)--(11.767,7.172);
\node[gp node right] at (2.424,7.172) {$b_2 = b_{0,1}$};
\gpcolor{color=gp lt color axes}
\gpsetlinetype{gp lt axes}
\gpsetdashtype{gp dt axes}
\gpsetlinewidth{0.50}
\draw[gp path] (2.608,8.441)--(11.947,8.441);
\gpcolor{color=gp lt color border}
\gpsetlinetype{gp lt border}
\gpsetdashtype{gp dt solid}
\gpsetlinewidth{1.00}
\draw[gp path] (2.608,8.441)--(2.788,8.441);
\draw[gp path] (11.947,8.441)--(11.767,8.441);
\node[gp node right] at (2.424,8.441) {$b_1 = b_{0,0}$};
\draw[gp path] (2.608,8.441)--(2.608,0.523)--(11.947,0.523)--(11.947,8.441)--cycle;
\node[gp node center,rotate=-270] at (0.276,4.482) {Interval Boundaries $b_{r}$};
\node[gp node center] at (7.277,0.215) {Index $r$};
\node[gp node right] at (10.479,8.107) {Exponential function};
\gpcolor{rgb color={0.580,0.000,0.827}}
\gpsetdashtype{dash pattern=on 10.00*\gpdashlength off 10.00*\gpdashlength on 10.00*\gpdashlength off 10.00*\gpdashlength on 10.00*\gpdashlength off 10.00*\gpdashlength }
\gpsetlinewidth{2.00}
\draw[gp path] (10.663,8.107)--(11.579,8.107);
\draw[gp path] (2.608,8.441)--(2.702,8.209)--(2.797,7.984)--(2.891,7.765)--(2.985,7.552)%
  --(3.080,7.344)--(3.174,7.143)--(3.268,6.947)--(3.363,6.756)--(3.457,6.570)--(3.551,6.390)%
  --(3.646,6.215)--(3.740,6.044)--(3.834,5.878)--(3.929,5.716)--(4.023,5.559)--(4.117,5.407)%
  --(4.212,5.258)--(4.306,5.114)--(4.400,4.973)--(4.495,4.836)--(4.589,4.703)--(4.683,4.574)%
  --(4.778,4.448)--(4.872,4.326)--(4.966,4.207)--(5.061,4.091)--(5.155,3.979)--(5.249,3.869)%
  --(5.344,3.763)--(5.438,3.659)--(5.532,3.559)--(5.627,3.461)--(5.721,3.365)--(5.815,3.273)%
  --(5.910,3.183)--(6.004,3.095)--(6.098,3.010)--(6.193,2.927)--(6.287,2.846)--(6.381,2.768)%
  --(6.476,2.691)--(6.570,2.617)--(6.664,2.545)--(6.759,2.475)--(6.853,2.407)--(6.947,2.340)%
  --(7.042,2.276)--(7.136,2.213)--(7.230,2.152)--(7.325,2.092)--(7.419,2.034)--(7.513,1.978)%
  --(7.608,1.923)--(7.702,1.870)--(7.796,1.819)--(7.891,1.768)--(7.985,1.719)--(8.079,1.672)%
  --(8.174,1.625)--(8.268,1.580)--(8.362,1.537)--(8.457,1.494)--(8.551,1.453)--(8.645,1.412)%
  --(8.740,1.373)--(8.834,1.335)--(8.928,1.298)--(9.023,1.262)--(9.117,1.227)--(9.211,1.193)%
  --(9.306,1.160)--(9.400,1.127)--(9.494,1.096)--(9.589,1.065)--(9.683,1.036)--(9.777,1.007)%
  --(9.872,0.979)--(9.966,0.952)--(10.060,0.925)--(10.155,0.899)--(10.249,0.874)--(10.343,0.850)%
  --(10.438,0.826)--(10.532,0.803)--(10.626,0.780)--(10.721,0.758)--(10.815,0.737)--(10.909,0.716)%
  --(11.004,0.696)--(11.098,0.677)--(11.192,0.658)--(11.287,0.639)--(11.381,0.621)--(11.475,0.604)%
  --(11.570,0.587)--(11.664,0.570)--(11.758,0.554)--(11.853,0.538)--(11.947,0.523);
\gpcolor{color=gp lt color border}
\node[gp node right] at (10.479,7.799) {Piecewise linear interpolations};
\gpcolor{rgb color={0.000,0.620,0.451}}
\gpsetdashtype{gp dt solid}
\draw[gp path] (10.663,7.799)--(11.579,7.799);
\draw[gp path] (2.608,8.441)--(2.646,8.379)--(2.683,8.318)--(2.721,8.256)--(2.759,8.195)%
  --(2.797,8.133)--(2.834,8.072)--(2.872,8.010)--(2.910,7.949)--(2.948,7.887)--(2.985,7.826)%
  --(3.023,7.764)--(3.061,7.703)--(3.099,7.641)--(3.136,7.580)--(3.174,7.518)--(3.212,7.457)%
  --(3.249,7.395)--(3.287,7.334)--(3.325,7.272)--(3.363,7.211)--(3.400,7.149)--(3.438,7.087)%
  --(3.476,7.026)--(3.514,6.964)--(3.551,6.903)--(3.589,6.841)--(3.627,6.780)--(3.665,6.718)%
  --(3.702,6.657)--(3.740,6.595)--(3.778,6.534)--(3.815,6.472)--(3.853,6.411)--(3.891,6.349)%
  --(3.929,6.288)--(3.966,6.226)--(4.004,6.165)--(4.042,6.103)--(4.080,6.042)--(4.117,5.980)%
  --(4.155,5.919)--(4.193,5.857)--(4.231,5.796)--(4.268,5.734)--(4.306,5.672)--(4.344,5.611)%
  --(4.381,5.549)--(4.419,5.488)--(4.457,5.426)--(4.495,5.365)--(4.532,5.303)--(4.570,5.242)%
  --(4.608,5.180)--(4.646,5.119)--(4.683,5.057)--(4.721,4.996)--(4.759,4.934)--(4.797,4.873)%
  --(4.834,4.811)--(4.872,4.750)--(4.910,4.688)--(4.947,4.627)--(4.985,4.565)--(5.023,4.504)%
  --(5.061,4.442)--(5.098,4.380)--(5.136,4.319)--(5.174,4.257)--(5.212,4.196)--(5.249,4.134)%
  --(5.287,4.073)--(5.325,4.011)--(5.363,3.950)--(5.400,3.888)--(5.438,3.827)--(5.476,3.765)%
  --(5.513,3.704)--(5.551,3.642)--(5.589,3.581)--(5.627,3.519)--(5.664,3.458)--(5.702,3.396)%
  --(5.740,3.335)--(5.778,3.273)--(5.815,3.212)--(5.853,3.150)--(5.891,3.089)--(5.929,3.027)%
  --(5.966,2.965)--(6.004,2.904)--(6.042,2.842)--(6.079,2.781)--(6.117,2.719)--(6.155,2.658)%
  --(6.193,2.596)--(6.230,2.535)--(6.268,2.473)--(6.306,2.412)--(6.344,2.350);
\draw[gp path] (5.098,3.771)--(5.142,3.743)--(5.186,3.714)--(5.230,3.685)--(5.274,3.657)%
  --(5.319,3.628)--(5.363,3.599)--(5.407,3.570)--(5.451,3.542)--(5.495,3.513)--(5.539,3.484)%
  --(5.583,3.456)--(5.627,3.427)--(5.671,3.398)--(5.715,3.369)--(5.759,3.341)--(5.803,3.312)%
  --(5.847,3.283)--(5.891,3.255)--(5.935,3.226)--(5.979,3.197)--(6.023,3.168)--(6.067,3.140)%
  --(6.111,3.111)--(6.155,3.082)--(6.199,3.054)--(6.243,3.025)--(6.287,2.996)--(6.331,2.968)%
  --(6.375,2.939)--(6.419,2.910)--(6.463,2.881)--(6.507,2.853)--(6.551,2.824)--(6.595,2.795)%
  --(6.639,2.767)--(6.683,2.738)--(6.727,2.709)--(6.771,2.680)--(6.815,2.652)--(6.859,2.623)%
  --(6.903,2.594)--(6.947,2.566)--(6.991,2.537)--(7.035,2.508)--(7.079,2.479)--(7.123,2.451)%
  --(7.167,2.422)--(7.211,2.393)--(7.255,2.365)--(7.300,2.336)--(7.344,2.307)--(7.388,2.278)%
  --(7.432,2.250)--(7.476,2.221)--(7.520,2.192)--(7.564,2.164)--(7.608,2.135)--(7.652,2.106)%
  --(7.696,2.077)--(7.740,2.049)--(7.784,2.020)--(7.828,1.991)--(7.872,1.963)--(7.916,1.934)%
  --(7.960,1.905)--(8.004,1.877)--(8.048,1.848)--(8.092,1.819)--(8.136,1.790)--(8.180,1.762)%
  --(8.224,1.733)--(8.268,1.704)--(8.312,1.676)--(8.356,1.647)--(8.400,1.618)--(8.444,1.589)%
  --(8.488,1.561)--(8.532,1.532)--(8.576,1.503)--(8.620,1.475)--(8.664,1.446)--(8.708,1.417)%
  --(8.752,1.388)--(8.796,1.360)--(8.840,1.331)--(8.884,1.302)--(8.928,1.274)--(8.972,1.245)%
  --(9.016,1.216)--(9.060,1.187)--(9.104,1.159)--(9.148,1.130)--(9.192,1.101)--(9.236,1.073)%
  --(9.281,1.044)--(9.325,1.015)--(9.369,0.986)--(9.413,0.958)--(9.457,0.929);
\draw[gp path] (8.211,1.498)--(8.255,1.486)--(8.299,1.475)--(8.343,1.463)--(8.387,1.452)%
  --(8.432,1.440)--(8.476,1.429)--(8.520,1.417)--(8.564,1.406)--(8.608,1.394)--(8.652,1.383)%
  --(8.696,1.371)--(8.740,1.360)--(8.784,1.348)--(8.828,1.337)--(8.872,1.325)--(8.916,1.314)%
  --(8.960,1.302)--(9.004,1.291)--(9.048,1.279)--(9.092,1.268)--(9.136,1.256)--(9.180,1.245)%
  --(9.224,1.233)--(9.268,1.222)--(9.312,1.210)--(9.356,1.199)--(9.400,1.187)--(9.444,1.176)%
  --(9.488,1.164)--(9.532,1.153)--(9.576,1.142)--(9.620,1.130)--(9.664,1.119)--(9.708,1.107)%
  --(9.752,1.096)--(9.796,1.084)--(9.840,1.073)--(9.884,1.061)--(9.928,1.050)--(9.972,1.038)%
  --(10.016,1.027)--(10.060,1.015)--(10.104,1.004)--(10.148,0.992)--(10.192,0.981)--(10.236,0.969)%
  --(10.280,0.958)--(10.324,0.946)--(10.368,0.935)--(10.413,0.923)--(10.457,0.912)--(10.501,0.900)%
  --(10.545,0.889)--(10.589,0.877)--(10.633,0.866)--(10.677,0.854)--(10.721,0.843)--(10.765,0.831)%
  --(10.809,0.820)--(10.853,0.808)--(10.897,0.797)--(10.941,0.785)--(10.985,0.774)--(11.029,0.763)%
  --(11.073,0.751)--(11.117,0.740)--(11.161,0.728)--(11.205,0.717)--(11.249,0.705)--(11.293,0.694)%
  --(11.337,0.682)--(11.381,0.671)--(11.425,0.659)--(11.469,0.648)--(11.513,0.636)--(11.557,0.625)%
  --(11.601,0.613)--(11.645,0.602)--(11.689,0.590)--(11.733,0.579)--(11.777,0.567)--(11.821,0.556)%
  --(11.865,0.544)--(11.909,0.533)--(11.945,0.523);
\gpcolor{rgb color={0.000,0.000,0.000}}
\gpsetpointsize{8.00}
\gppoint{gp mark 1}{(2.608,8.441)}
\gppoint{gp mark 1}{(3.386,7.172)}
\gppoint{gp mark 1}{(4.165,5.903)}
\gppoint{gp mark 1}{(4.943,4.634)}
\gppoint{gp mark 1}{(5.721,3.365)}
\gppoint{gp mark 1}{(6.499,2.858)}
\gppoint{gp mark 1}{(7.278,2.350)}
\gppoint{gp mark 1}{(8.056,1.843)}
\gppoint{gp mark 1}{(8.834,1.335)}
\gppoint{gp mark 1}{(9.612,1.132)}
\gppoint{gp mark 1}{(10.391,0.929)}
\gppoint{gp mark 1}{(11.169,0.726)}
\gppoint{gp mark 1}{(11.947,0.523)}
\gpcolor{color=gp lt color border}
\gpsetlinewidth{1.00}
\draw[gp path] (2.608,8.441)--(2.608,0.523)--(11.947,0.523)--(11.947,8.441)--cycle;
\gpdefrectangularnode{gp plot 1}{\pgfpoint{2.608cm}{0.523cm}}{\pgfpoint{11.947cm}{8.441cm}}
\end{tikzpicture}